\RequirePackage[thmmarks]{ntheorem}
\makeatletter
\renewtheoremstyle{plain} 
{\item[\hskip\labelsep \theorem@headerfont ##1\ \textup{##2}\theorem@separator]} 
{\item[\hskip\labelsep \theorem@headerfont ##1\ \textup{##2}\ (##3)\theorem@separator]}
\makeatother

\let\latexdocument\document
\let\latexarabic\arabic

\documentclass[lineo]{article}
\let\document\latexdocument
\let\arabic\latexarabic

\usepackage{amsmath}
\usepackage{amssymb,amsfonts,bbm,graphics,rotating,graphicx,stmaryrd, natbib,url, xcolor,ulem}
\usepackage{times}
\usepackage{bm}
\usepackage{natbib}

\graphicspath{{./art/}}

\usepackage[plain,noend]{algorithm2e}

\makeatletter
\renewcommand{\algocf@captiontext}[2]{#1\algocf@typo. \AlCapFnt{}#2} % text of caption
\def\@algocf@capt@plain{top}
\renewcommand{\algocf@makecaption}[2]{%
	\addtolength{\hsize}{\algomargin}%
	\sbox\@tempboxa{\algocf@captiontext{#1}{#2}}%
	\ifdim\wd\@tempboxa >\hsize%     % if caption is longer than a line
	\hskip .5\algomargin%
	\parbox[t]{\hsize}{\algocf@captiontext{#1}{#2}}% then caption is not centered
	\else%
	\global\@minipagefalse%
	\hbox to\hsize{\box\@tempboxa}% else caption is centered
	\fi%
	\addtolength{\hsize}{-\algomargin}%
}
\makeatother

\newcommand{\bomega}{{\boldsymbol\omega}}
\renewcommand{\L}{\mathbb{L}}
\newcommand{\vvvert}{\rvert\hspace{-0.12em}\rvert\hspace{-0.12em}\rvert}

\newcommand{\R}{\ensuremath{\mathbb{R}}}
\newcommand{\C}{\ensuremath{\mathcal{C}}}

\newcommand{\var}{\mbox{Var}}

\newcommand{\argmin}{\mathop{\mathrm{arg\,\min}}}

\newcommand{\1}{\mathbf{1}}

\newcommand{\MD}[1]{\text{M}_{\Delta}}

\newtheorem{theorem}{Theorem}
\newtheorem{definition}{Definition}
\newtheorem{remark}{Remark}
\newtheorem{proposition}{Proposition}
\newtheorem{corollary}{Corollary}
\newtheorem{assumption}{Assumption}
\newtheorem*{proof}{Proof}
\newtheorem{lemma}{Lemma}
\usepackage{authblk}
\usepackage{xr}
\providecommand{\keywords}[1]
{
  \small	
  \textbf{\textit{Keywords---}} #1
}

\begin{document}

%% The left and right page headers are defined here:
\markboth{R. Belhakem et~al.}{}

%% Here are the title, author names and addresses
\title{Minimax estimation of Functional Principal Components from noisy discretized functional data}

\author[1]{Ryad Belhakem \thanks{belhakem@ceremade.dauphine.fr}}
\affil[1,3,4]{CEREMADE, CNRS, Universit\'e Paris-Dauphine, Universit\'e PSL, 75016 PARIS, FRANCE}

\author[2]{Franck Picard \thanks{franck.picard@ens-lyon.fr}}
\affil[2]{Centre National de la Recherche Scientifique, Laboratoire de Biologie et Mod\'elisation de la Cellule\\
46, all\'ee d'Italie
69007 Lyon}

\author[3]{Vincent Rivoirard \thanks{rivoirard@ceremade.dauphine.fr}}

\author[4]{Angelina Roche 
\thanks{roche@ceremade.dauphine.fr}}

\maketitle

\newpage

\begin{abstract}
Functional Principal Component Analysis is a reference method for dimension reduction of curve data. Its theoretical properties are now well understood in the simplified case where the sample curves are fully observed without noise. However, functional data are noisy and necessarily observed on a finite discretization grid.
Common practice consists in smoothing the data and then to compute the functional estimates, but the impact of this denoising step on the procedure's statistical performance are rarely considered. Here we prove new convergence rates for functional principal component estimators. We introduce a double asymptotic framework: one corresponding to the sampling size and a second  to the size of the grid. We prove that estimates based on projection onto histograms show optimal rates in a minimax sense. Theoretical results are illustrated on simulated data and the method is  applied to the visualization of genomic data. 
\end{abstract}


\newpage

\section{Introduction}

Functional Data Analysis (FDA) is a statistical framework dedicated to curve data that are supposed to be the realizations of random functions \citep{FV06,Ram,FR11}. %The infinite-dimensional nature of the data under study has to be taken into account in order to develop efficient estimation procedures. 
Hence, in this framework, the infinite-dimensional nature of the process that generated the data is central to develop efficient estimation procedures. Functional Principal Components Analysis (fPCA) is a common method to reduce dimensionality of curve data and has been considered either as an exploration tool \citep{Ram} or as a pre-processing step for many statistical procedures \citep{Kalogridis,Goode,T.V.E,Wookyeong,won,seb}. In most theoretical works (see e.g. \cite{CJ10,CY12,Mas}) it is assumed that the functional data $Z_i(t)$ is observed for individual $i$ at all points $t$ in an interval (e.g. $[0,1]$). However, in practice, the $Z_i$'s are observed on a finite grid $t_{i,0},\hdots,t_{i,p_i-1}$ and can be corrupted by noise. Few theoretical works in the literature of FDA study the effect of the sampling scheme on the performance of the estimation procedures. 

Here, we consider the case where data are observed on a fixed and regular design (fixed grid), and we focus on the estimation of the elements of the functional principal components, from a non-asymptotic point of view. The case of a fixed regular grid actually corresponds to a large number of applications in FDA, for instance electricity consumption curves~\citep{DGP20}, temperature or precipitation curves~\citep{Ram}, or spectrometric datasets~\citep{Pham+10}) to name a few. 
\textcolor{black}{More precisely, in this paper, we observe  $(Y_i(t_h))_{i=1,\hdots,n, h=0,\hdots,p-1}$
generated from the following statistical model 
\begin{equation*}
	Y_i(t_h)=Z_i(t_h)+\varepsilon_{i,h}, \quad \quad i=1,\hdots,n,
\end{equation*}
where $\{\varepsilon_{i,h}\}_{i=1,\hdots,n; h=0,\hdots,p-1}$ is an i.i.d. sequence of centered  errors, the $Z_i$'s are i.i.d. random elements of the space of continuous functions on the interval $[0,1]$ and $t_h=h/(p-1)$, $h=0,\hdots,p-1$. See Section~\ref{setting} for the introduction of the precise setting.
}

Two statistical frameworks are often considered: longitudinal data analysis (LDA) or FDA and it seems important to clarify at this stage the links and differences between these two approaches.  %In LDA, the sampling points depends from an individual to the other and are usually small. 
In LDA, data are often observed at random, hence the sampling points are random and depend on each individual. From a theoretical point of view, the number of sampling points can be supposed bounded, and results concerning convergence rates only involve the number of individuals $n$. There have been a lot of methodological and theoretical works in the case where the observations are observed on a random grid with a small number of observations per subject (sparse longitudinal data), we refer to \cite{YMW05,Hall,DMT18,ZSHJ22} and references therein. In this paper, we consider the case where the number of sampling points is usually considered to be large and shared by all curves, which is frequently the case in FDA and motivates the characterization of a double asymptotic in $n$ and $p$.

At first sight, inference for fPCA  and classical PCA on the $n\times p$ matrix of observations is comparable. Following this idea, when $p \gg n $, functional PCA would be confronted to inconsistency problems as in standard multivariate PCA (see \citealt{Johnstone} and references therein). However, this is counter-intuitive in the functional framework: when $p$ increases, more and more information is recorded on the underlying process, which should improve statistical performance. Indeed, in the continuous non-noisy case, where $Z_i(t)$ is observed at all points $t$, corresponding to $p=+\infty$ and $\sigma=0$, it is known from the works of \citet{DPR82} that the estimation of the principal components is consistent. Moreover, an optimal $n^{-1}$ parametric rate (up to a logarithmic factor) can be achieved for the risk associated to the $\mathbb L^2$-error \citep[Theorem 4.5, p. 106]{Bosq} or to the operator norm of the projector \citep{Mas}. Considering our data as multivariate data would resume to ignore the underlying regularity of the processes $Z_i$ and the fact that, when $p$ is large, $Z_i(t_h)$ is close to $Z_i(t_{h+1})$. Then, specific attention should be paid.

From a theoretical perspective, the main challenge is to assess the rates of convergence of the estimators of the elements of the principal components basis, in a very specific framework. FDA combines two very different convergence settings: a first one associated with the sampling of $n$ independent processes $Z_1,\hdots,Z_n$, and a non-parametric setting since data are functions, here observed at $p$ points. 
\cite{Hall} investigate the estimation of the elements of the fPCA basis in the case where the number of discretization points by individual is bounded and the grid is random (i.e. $p_i\leq p$ with $p$ fixed and $n\to+\infty$). They obtain non-parametric rates for a kernel smoothing estimator which are optimal under the assumption that the function to estimate is exactly two-times differentiable. However, the authors themselves point out that their results are no longer valid in the case of a fixed regular grid, where consistent estimation is not possible when $p$ is fixed. In the context of the estimation of the mean function of a functional data sample, \citet{CY11} found that the optimal rates of convergence are completely different if we consider a fixed grid or a random grid. It appears that, in the case of a fixed grid, the minimax estimation rates of the function principal components basis remain unknown. Hence we propose to investigate the joint impact of noise and discretization (sampling scheme) on the estimation of the eigenelements of the covariance operator.

In addition, a common practice is to first smooth the data, usually by projecting it into a splines basis with a roughness penalty or via kernel smoothing (see e.g. \citealt{Ram}). However, the statistical implications of the smoothing step are rarely debated, whereas it raises some concern, mainly related to the level of regularity of the underlying process versus the choice of the smoothing basis, and the capacity of distinguishing noise from signal through this method. \citep[Section 2.2, p. 2336]{CY11} point out that, in the case of mean estimation, there is no benefit from smoothing in terms of convergence rates when the observation grid is fixed. This implies that usual splines or kernel smoothing step can lead to suboptimal estimators if the smoothing parameter is not well chosen (see Theorem 2.2 of \cite{CY11}).

To fully understand the statistical complexity of functional principal component analysis, it is necessary to compute the minimax rates of estimation and to compare it with the parametric bounds obtained by \citet{Mas}. Upper bounds in the case of noisy discretized data have also been proposed \citep{BX15,DP19}. \cite{BX15} established results under strong conditions of the eigenvalues of the operator. \cite{DP19} studied a generalization to heterogeneous noise with possible time dependency at the price of two strong assumptions: analyticity of the eigenfunctions and finite rank of the covariance operator  of the signal; the achieved rate is then $n^{-1}+p^{-2}$.

Here we study convergence rates for the estimation of the eigenelements of the covariance operator $\Gamma$ under a mild regularity assumption on the process $Z$. Denoting by $\alpha$ this regularity, our assumption is equivalent to assuming that the kernel $K$ is a bivariate $\alpha$-H\"older continuous function. Under a moment assumption for process $Z$, we obtain rates of the form
\[
n^{-1}+p^{-2\alpha}.
\]
These rates, which are new, are, moreover, optimal in the minimax sense for the estimation of the first eigenfunction (we prove a lower bound). We illustrate these rates in practice on simulations.

These rates tell us a lot about the behavior of the estimated eigenfunctions under the double asymptotic in $p$ and $n$. When $p$ is large compared to $n^{1/(2\alpha)}$, we find the optimal parametric rate $n^{-1}$ obtained by \citet{DPR82,Mas} when the curves $Z_i$ are fully observed and without noise. Moreover, even though the problem is intrinsically non-parametric, and in the presence of noisy observations,  the simple estimator obtained by projection on the $p$-bins histogram system reaches the optimal minimax rate. Therefore, we do not need regularization, which may be counter-intuitive. The knowledge of $\alpha$ is also not necessary. These results are confirmed by the simulation study we have conducted which also suggests that the same conclusion applies for the estimation of the eigenvalues for which we do not know so far whether the rate obtained is minimax optimal. Lastly, at the end of Section~\ref{sec:upperspecific}, we discuss our results and compare them with the most recent ones.  Our results show that the underlying regularity of the data plays a special role in the theoretical developments, with important implications in practice. In our setting with equispaced deterministic observations, the sampling scheme cannot be too sparse to obtain parametric rates and even consistency. This is a main difference with the random sampling scheme considered by \citep{Hall}. We complete our theoretical and empirical study by two original applications of functional principal component analysis, on single-cell expression data analysis to characterize the immune response to viral infection, and to genomic data for the characterization of replication origins along the human genome with respect to the spatial distribution of particular sequence motifs called G-quadruplexes \citep{ZZH20}.

\textbf{Notations:} We denote $\|\cdot\|$ the $\L^2$-norm associated with the scalar product $\langle\cdot,\cdot\rangle$ and $\|\cdot\|_{\ell_2}$ the $\ell_2$-norm for a vector. For any continuous operator $T$ on $\L_2[0,1]$, we denote $\vvvert T\vvvert $ the operator norm of $T$ associated to $\|\cdot\|$ and defined by $\vvvert T\vvvert =\sup_{f\in\L_2[0,1], \|f\|=1}\|Tf\|$.
For $P$ a probability measure, we denote $E$ the associated expectation.  We denote $P_Z$ the distribution of the process $Z$ and $E_Z$ the associated expectation. The set of continuous functions on $[0,1]$ is denoted $\C^0$. We adopt the following notation: for two sequences $a=(a_{n,p})_{n,p\geq 1}$, $b=(b_{n,p})_{n,p\geq 1}$ of real fixed quantities or random variables, we denote $a\lesssim b$ if there exists \textcolor{black}{a universal positive constant $c$} such that $a_{n,p}\leq c b_{n,p}$ a.s. for all $n,p\geq 1$. We define $\text{sign}(u)=\1_{\{u\geq 0\}}-\1_{\{u<0\}}$ for any $u\in\R$.

%%%%%%%%%%%%%%%%%%%%%%%%%%%%%%%%
%%%%%%%%%%%%%%%%%%%%%%%%%%%%%%%%
\section{Functional PCA for discretely observed random functions}\label{sec:estim}
\textcolor{black}{Section~\ref{setting} introduces the setting considered along this paper. We then define in Section~\ref{section:estimator} the statistical procedure studied in Sections~\ref{sec:theory} and \ref{sec:simus}.}
%%%%%%%%%%%%%%%%%%%%%%%%%%%%%%%%
\textcolor{black}{\subsection{Setting of the paper}\label{setting}
Along this paper, as mentioned in Introduction, we consider the following statistical model 
\begin{equation}\label{model}
	Y_i(t_h)=Z_i(t_h)+\varepsilon_{i,h}, \quad \quad i=1,\hdots,n,\quad h=0,\hdots,p-1,
\end{equation}
where 
\begin{equation}\label{th}
t_h=\frac{h}{p-1},\quad h=0,\hdots,p-1,
\end{equation}
the $Z_i$'s are i.i.d. random elements of the space of continuous functions on the interval $[0,1]$ with the same distribution as $Z$ such that 
\begin{equation}\label{moments-Z}
E(\|Z\|^2)<\infty.
\end{equation}
In Equation~\eqref{model}, we assume that 
$\{\varepsilon_{i,h}\}_{i=1,\hdots,n; h=0,\hdots,p-1}$ is an i.i.d. sequence of centered Gaussian errors with variance $\sigma^2$ and assume that the $\varepsilon_{i,h}$'s are independent from the $Z_i$'s. Observe that the Gaussian assumption, which avoids tedious technicalities, can be relaxed. See subsequent Remarks~\ref{rem1:Gaussian} and \ref{rem2:Gaussian}.
Under \eqref{moments-Z}, we can introduce the covariance operator $\Gamma$ defined by
$$
\Gamma(f)(\cdot) = E (\langle f,Z\rangle Z(\cdot) ), \quad f \in {\mathbb L}^2. 
$$
Let  $\eta^*=\{ \eta^*_d, d\in \mathbb{N}^*\}$ the eigenfunctions of $
\Gamma$ and $\mu^*=\{\mu^*_d,d \in \mathbb{N}^* \,; \, \mu^*_1>\mu^*_2> \hdots\}$ the associated eigenvalues. We assume that these eigenvalues are distinct. The Karhunen-Lo\`eve representation of $Z$ is then \citep{Bosq}:
\begin{equation}\label{eq:KL}
Z = \sum_{d \in \mathbb{N}^*} \zeta^*_d \mu^{*1/2}_d\eta_d^*,
\end{equation}
where  $\zeta^*=\{ \zeta^*_d, d\in \mathbb{N}^*\}$ is a sequence of non-correlated centered random variables of variance~1, usually called the principal components scores and the family $(\eta_j^*)_{j\geq 1}$ is an orthonormal basis of $\mathbb L^2$.
 Then, for any integer $\mathcal D$, the best $\mathcal D$-dimensional approximation  of process $Z$ is
spanned by the first $\mathcal D$ eigenfunctions. Our aim is to provide estimators of these eigenelements based on noisy discretized data, and to assess their statistical performance.}
%%%%%%%%%%%%%%%%%%%%%%%%%%%%%%%%
\subsection{Estimator of the covariance operator}\label{section:estimator}
In \citep{DPR82,Bosq,Ram,Hall}, the $Z_i(t)$'s are observed for all $t\in[0,1]$ without noise and the estimator of $\eta_d^*$ is the eigenfunction $\widehat\eta_d$ associated to the $d$-th largest eigenvalue of the empirical covariance operator
$$
\widehat\Gamma(f)(\cdot)  = \frac1n\sum_{i=1}^n\langle f,Z_i\rangle Z_i(\cdot) , \quad f \in \L^2.
$$ 
However, when process $Z$ is observed on a grid, the empirical covariance operator $\widehat\Gamma$ can not be calculated and must be approximated. 
%\enleverp{Following the usual approach of functional data analysis , the estimation procedure for functional principal component analysis consists in estimating the covariance operator $\Gamma$ from the data:
%\begin{eqnarray*}
%\Gamma(f)(\cdot) &=& \int_0^1 K(s, \cdot)f(s)ds, \quad  f \in \L^2,
%\end{eqnarray*}
%that \enlever{is well defined provided $E(\|Z\|^2)<+\infty$, which is assumed in the following.}}
%\enleverp{Then, a natural approach consists in estimating $K$ in a first step using the empirical covariance kernel:
%$$
%\widehat K(s,t)=\frac1n\sum_{i=1}^n Z_i(t)Z_i(s), \quad  (s,t) \in [0,1]^2.
%$$ }
%\enleverp{Since this estimator cannot be calculated directly from the data, w}We introduce $(\phi_\lambda)_{\lambda \in \Lambda_D}$, an orthonormal system of $\L^2$ with $\Lambda_D$ a finite set of size $D$. \enlever{In the following, we will consider histograms and Haar wavelets. }

Then, in the setting of Model~\eqref{model}, we first reconstruct the observed curves on the entire interval and we define, for $i = 1, \hdots, n,$
\[
\widetilde Y_i(t) = \sum_{\lambda\in\Lambda_D}\widetilde y_{i,\lambda} \phi_\lambda(t),
\quad 
\widetilde y_{i,\lambda}=\frac1p\sum_{h=0}^{p-1}Y_i(t_h)\phi_\lambda(t_h), \quad t \in [0,1],
\]
where $\{\phi_\lambda,\lambda\in\Lambda_D\}$ is an orthonormal system of $\mathbb L^2([0,1])$ of cardinality $D\geq 1$ and $\widetilde y_{i,\lambda}$ an approximation of $\langle Y_i,\phi_\lambda\rangle$. Similarly, we define $\widetilde Z_i(t)$, $\widetilde z_{i,\lambda}$, $\widetilde E_i(t)$, $\widetilde \varepsilon_{i,\lambda}$ by replacing $Y_i(t_h)$ in the previous expressions by $Z_i(t_h)$, and $\varepsilon_{i,h}$.  %\enleverp{A natural estimator of the covariance kernel $K$ is then
%\begin{equation*}%\label{eq:Kestim}
%\widehat{K}_\phi(s,t)=\frac1n\sum_{i=1}^n \widetilde Y_i(t)\widetilde Y_i(s),\quad  (s,t) \in [0,1]^2,
%\end{equation*}
%and the covariance operator $\Gamma$ can be estimated by
%\begin{equation*}
%\widehat{\Gamma}_\phi(f)(\cdot) = \int_0^1 \widehat{K}_\phi(s,\cdot)f(s)ds, \quad f \in \L^2.
%\end{equation*}
%Since $\widehat{K}_\phi$ is symmetric, the operator $\widehat{\Gamma}_\phi$ is self-adjoint and is also finite-rank since $\operatorname{Im}(\widehat{\Gamma}_\phi)\subset\text{span}(\widetilde Y_1,\hdots,\widetilde Y_n)$. Therefore, $\widehat{\Gamma}_\phi$ is a compact operator.
%From the spectral theorem, we know that there exists a $\L^2$-basis of eigenfunctions of $\widehat{\Gamma}_\phi$, denoted by $ \widehat{\eta}_\phi = \{ \widehat{\eta}_{\phi,d}, d \in \mathbb{N}^* \}$, with associated eigenvalues $ \widehat{\mu}_\phi = \{ \widehat{\mu}_{\phi,d}, d \in \mathbb{N}^*;  \widehat{\mu}_{\phi,1}\geq\widehat{\mu}_{\phi,2} \geq\hdots \}$. We then obtain estimates of the principal components that are analyzed in the minimax setting.}

A natural estimator of the covariance operator is then 
\begin{equation*}
\widehat{\Gamma}_\phi(f)(\cdot) = \frac1n\sum_{i=1}^n\langle f,\widetilde Y_i\rangle\widetilde Y_i(\cdot), \quad f \in \L^2.
\end{equation*}
It is easily seen from the definition above that the operator $\widehat{\Gamma}_\phi$ is self-adjoint. It is also finite-rank hence compact. Then, by the diagonalization theorem for self-adjoint compact operators~\citep[see][Theorem 6.11, p. 167]{Brezis11}, there exists a basis $(\widehat{\eta}_{\phi,d})_{d\geq 1}$ of $\mathbb L^2$ made of eigenfunctions of $\widehat{\Gamma}_\phi$. In the following, we study the $\mathbb L^2$-risk of the estimator $\widehat{\eta}_{\phi,d}$ for $d=1,\hdots,\mathcal D$.
%%%%%%%%%%%%%%%%%%%%%%%%%%%%%%%%
%%%%%%%%%%%%%%%%%%%%%%%%%%%%%%%%
\section{Minimax rates of the eigenfunction estimator}\label{sec:theory}
\textcolor{black}{Along this section, we consider the setting of Section~\ref{setting}.}
%%%%%%%%%%%%%%%%%%%%%%%%%%%%%%%%
\subsection{Smoothness class for the functional curve $Z$}\label{sec:regularity}
Minimax rates of convergence depend on the underlying smoothness of the process of interest. In the sequel, for any $\alpha\in(0,1]$ and $L>0$ we consider the regularity class
\begin{eqnarray*}
\mathcal R_\alpha(L)&=&\Big\{ P, \text{
probability measure on }\C^0\text{ such that } \nonumber\\
&& \hspace{1.5cm}\int_{\C^0}\{z(t)-z(s)\}^2dP(z)\leq L|t-s|^{2\alpha}, \quad (s,t) \in[0,1]^2
\Big\}. 
%\label{hyp:regZ}
\end{eqnarray*}
This regularity set is natural. Indeed, we can for instance remark that $P_Z$, the distribution of $Z$, satisfies
\[
P_Z\in  \mathcal R_\alpha(L)\Leftrightarrow E_Z[\{Z(t)-Z(s)\}^2]\leq L|t-s|^{2\alpha},\quad  (s,t)\in[0,1]^2.
\]
This condition can be seen as a regularity assumption on the covariance kernel 
\begin{equation*}
	K(s,t)= E\big\{Z(s)Z(t)\big\}, \quad (s,t) \in [0,1]^2.
\end{equation*}
Indeed, our regularity condition $E(\|Z\|^2)<+\infty$ combined with the condition $P_Z\in  \mathcal R_\alpha(L)$ imply that kernel $K$ is  bounded  
\[\|K\|_\infty=\sup_{(s,t)\in[0,1]^2}|K(s,t)|<\infty,\]
%\vinc{donner une vraie borne ?\\} 
and is an $\alpha$-H\"older continuous function. More precisely, for any $(s,s',t,t')\in[0,1]^4$,
\begin{equation}\label{Kernelregularity}
P_Z\in  \mathcal R_\alpha(L)\Rightarrow|K(s,t)-K(s',t')|\leq (\|K\|_\infty L)^{1/2}\left(|s-s'|^\alpha+|t-t'|^\alpha\right).
\end{equation}
Conversely, if $K$ is a bivariate $\alpha$-H\"older continuous function, we know that there exists $L'>0$ such that 
\[
|K(s,t)-K(s',t')|\leq L'\left(|s-s'|^2+|t-t'|^2\right)^{\alpha/2}.
\]
Then 
\[
 E_Z[\{Z(t)-Z(s)\}^2]=K(s,s)-2K(s,t)+K(t,t)\leq 2L'|s-t|^\alpha,
\]
and $P_Z\in\mathcal R_\alpha(2L')$.

Classical Gaussian processes belong to $\mathcal R_\alpha(L)$ for $\alpha$ and $L$ well chosen. For instance, if $Z$ is a standard Brownian motion or a Brownian bridge then $P_Z\in\mathcal R_{1/2}(1)$. More generally, fractional Brownian motions with Hurst exponent $\alpha$ and Hurst index $C_\alpha$ belong to $\mathcal R_\alpha( C_\alpha)$. If $Z$ is an Ornstein-Uhlenbeck process, its covariance function is $K(s,t)=\exp(-|t-s|/2)$, then it verifies 
\[
E_Z[\{Z(t)-Z(s)\}^2]=2\big(1-e^{-|t-s|/2}\big)\leq |t-s|, \quad (s,t)\in {\mathbb R}^2,
\]
which implies $P_Z\in\mathcal R_{1/2}(1)$. We refer to \cite{Lifshits95} for the precise definitions and properties of these processes.\\
\begin{remark}\label{rk:L}We also remark that $L$ depends on the eigenvalues sequence $(\mu_d^*)_{d\geq 1}$. Indeed, suppose e.g. that, for all $d\geq 1$, the eigenfunction $\eta_d^*$ is $\alpha$-holdérian (i.e. there exists $L_d>0$ such that $|\eta_d^*(t)-\eta_d^*(s)|\leq L_d|t-s|^\alpha$, for all $t,s\in[0,1]$), then, if $\sum_{d\geq 1}\mu_d^*L_d^2<+\infty$, from the Karhunen-Loève decomposition~\eqref{eq:KL},%\\ we can write,
%\begin{eqnarray*}
%E_Z[{Z(t)-Z(s)}^2]&=&E_Z\left[\left\{\sum_{d\geq 1}\zeta_d^*\mu_d^{*1/2}(\eta_d^*(t)-\eta_d^*(s))\right\}^2\right]\nonumber\\
%&\leq& E_Z\left[\left\{\sum_{d\geq 1}L_d\zeta_d^*\mu_d^{*1/2}\right\}^2\right]\times |t-s|^{2\alpha}
%\end{eqnarray*}
%now remark that $\sum_{d\geq 1}L_d\zeta_d^*\mu_d^{*1/2}$ is a centered random variable, with variance $\sum_{d\geq 1}L_d^2\mu_d^*$ (since we recall that the $\zeta_d^*$'s are uncorrelated. Then $Z\in\mathcal R_{\alpha}(L)$ with $L=\sum_{d\geq 1}L_d^2\mu_d^*$.
%\\
we can write, since the $\zeta_d^*$'s are centered and uncorrelated,
\begin{eqnarray*}
E_Z[\{Z(t)-Z(s)\}^2]&=&E_Z\Bigg[\Big\{\sum_{d\geq 1}\zeta_d^*\mu_d^{*1/2}(\eta_d^*(t)-\eta_d^*(s))\Big\}^2\Bigg]\nonumber\\
&=&\sum_{d\geq 1}\mu_d^*(\eta_d^*(t)-\eta_d^*(s))^2\\
&\leq&\sum_{d\geq 1}\mu_d^*L_d^2|t-s|^{2\alpha}.
\end{eqnarray*}
Then $Z\in\mathcal R_{\alpha}(L)$ with $L=\sum_{d\geq 1}L_d^2\mu_d^*$.
\end{remark}
%%%%%%%%%%%%%%%%%%%%%%%%%%%%%%%%
\subsection{Lower bound}\label{sec:lowerbound}
The lower bound of the risk for estimating  eigenfunctions can be viewed as a benchmark to achieve. We focus on the first eigenfunction, but a similar result, though more technical, could be obtained for the other eigenfunctions. However, since the estimation of higher order eigenfunctions is a more complex statistical problem, it seems intuitively reasonable to us that the lower bound on these eigenfunctions is (at worst) of the same order.
\begin{theorem}\label{thm:borne_inf_fp}
Let $\alpha\in(0,1]$ and $L>0$.  \textcolor{black}{Assume that the rank of the covariance operator $\Gamma$ is larger than 2.} Then, for any $n\geq 1$ and $p\geq 1$, we have:
\[
\inf_{\widehat\eta_1}\sup_{P_Z\in\mathcal R_\alpha(L)}E(\|\widehat\eta_1-\eta_1^*\|^2)\geq c \Big(p^{-2\alpha}+n^{-1}\Big),
\] 
where $c$ is a positive constant depending on $L$, $\alpha$ and $\sigma$ and the infimum is taken over all estimators i.e. all measurable functions of the observations $\{Y_i(t_h), h=0,\hdots,p-1, i=1,\hdots,n\}$. 
\end{theorem}
Theorem~\ref{thm:borne_inf_fp} is obtained by combining Propositions~\ref{prop:borne_inf_fpn} and ~\ref{prop:borne_inf_fpp} stated in Appendix~\ref{proof:A}.% Both propositions rely on particular Gaussian processes whose parameters are chosen so that the associated distributions belong to $\mathcal R_\alpha(L)$. \textcolor{red}{However, since the supremum is evaluated on the whole set $\mathcal R_\alpha(L)$, this Gaussian assumption is not necessary in the statement of Theorem~\ref{thm:borne_inf_fp}. Indeed,...}

Proposition~\ref{prop:borne_inf_fpn} provides the parametric rate $n^{-1}$, which is expected in our setting where we observe $n$ curves. This rate has been proven to be optimal, up to logarthmic terms, in the case where the curves is supposed to be observed at all points (see e.g.~\citealt{Mas}). More precisely, %assuming the very mild assumption $p\geq 4$, 
Proposition~\ref{prop:borne_inf_fpn} gives
\[
\inf_{\widehat\eta_1}\sup_{P_Z\in\mathcal R_\alpha(L)}\mathbb E[\|\widehat\eta_1-\eta_1^*\|^2]\geq c_1n^{-1},\quad \inf_{\widehat\eta_2}\sup_{P_Z\in\mathcal R_\alpha(L)}\mathbb E[\|\widehat\eta_2-\eta_2^*\|^2]\geq c_1n^{-1},
\] 
where $c_1$ depends on $L$, $\alpha$ and $\sigma$. The key point to establish the lower bound is the explicit form of the Kullback-Leibler divergence between two Gaussian distributions. The first two eigenvalues introduced in the building of models of the proof of Proposition~\ref{prop:borne_inf_fpn} provide a constant spectral gap (i.e. $\mu_2^*-\mu_1^*$ is a constant independent of both $n$ and $p$). 
 
The lower bound of the minimax risk by $p^{-2\alpha}$ relies on the construction of two processes $Z_0,Z_1\in \mathcal R_\alpha(L)$ with first eigenfunctions distant of $p^{-\alpha}$ from each other and such that $Z_0(t_h)=Z_1(t_h)$ almost surely for all $h=0,\hdots,p-1$ (see Appendix~\ref{proof:borneinfp} and in particular Equation~\eqref{Z:def}).

Observe that if $p$, the number of observations per individual, is bounded, then rates cannot go to 0. In particular, consistency cannot be achieved by any estimate in our statistical model if $p$ is a constant. Our results corroborate arguments of the discussion section 3.2 of \citet{Hall}. Parametric rates can be achieved only if $p$ is large enough, namely larger than $n^{1/(2\alpha)}$.

We can remark that the constant $L$ appearing in the regularity class is strongly linked with the eigenvalues sequence $(\mu_d^*)_{d\geq 1}$. This can be seen via Remark~\ref{rk:L} but also via the proofs of Propositions~\ref{prop:borne_inf_fpn} and~\ref{prop:borne_inf_fpp} where the first eigenvalues of the processes we construct  are upper-bounded, up to multiplicative constants, by $L$.
%%%%%%%%%%%%%%%%%%%%%%%%%%%%%%%%
\subsection{General upper bounds}\label{sec:upperbound}

We now derive upper bounds for estimates $\widehat{\eta}_{\phi,d}$. For this purpose, we set 
\begin{equation*}%\label{margin1}
b_1=8(\mu_1^*-\mu_2^*)^{-2}
\end{equation*}
and for $d= 2,\hdots,\mathcal  D$,
\begin{equation*}%\label{margin}
b_d=8/\min(\mu_d^*-\mu_{d+1}^*,\mu_{d-1}^*-\mu_d^*)^2.
\end{equation*}
Since we supposed that all the true eigenvalues $\mu_d^*$'s are distinct, the quantities $b_d$'s are well defined and finite.

The eigenfunction $\eta_d^*$ being defined up to a sign change ($-\eta_d^*$ is also an eigenfunction associated to the eigenvalue $\mu_d^*$), we cannot assess our procedure by using the classical risk $E(\|\widehat{\eta}_{\phi,d}-\eta_d^*\|^2)$. Following \citet{Bosq}, we evaluate the risk of
\[\eta_{\pm,d}^*=\text{sign}(\langle\widehat{\eta}_{\phi,d},\eta_d^*\rangle)\times\eta_d^*.\] 
We consider the following mild assumption on the 4th moment of the vector 
\[{\bf Z}=\{Z(t_0),\ldots,Z(t_{p-1})\}^T.\]
\textcolor{black}{\begin{assumption}\label{Ass:4}
We assume that there exists $C_1>0$ such that 
\begin{equation}\label{hyp-4}
E\{(v^T{\bf Z})^4\}\leq C_1[E\{(v^T{\bf Z})^2\}]^2, \quad v\in\R^p.
\end{equation}
\end{assumption}
Assumption~\ref{Ass:4} ensures a control of the fourth moment of $\tilde z_{1,\lambda}$. It is satisfied with $C_1=3$ if ${\bf Z}$ is Gaussian.} Then we obtain the following result: 
\begin{theorem}\label{theo:GeneralEsp}
Let $d$ be fixed. Under Assumption~\ref{Ass:4}, we have
\begin{eqnarray*}
 E(\|\widehat{\eta}_{\phi,d}-\eta_{\pm,d}^*\|^2)&\leq& 5b_d \left[\vvvert\Pi_D\Gamma\Pi_D-\Gamma\vvvert^2+\frac{\max(C_1+3;6)}{n}\left\{\sum_{\lambda\in\Lambda_D}\Big(\sigma_\lambda^2+s_\lambda^2\Big)\right\}^2\right.\\
   &&\hspace{5cm} \left.+A_p^{(K)}(\phi,D)+A_p^{(\sigma)}(\phi,D)+\frac{\sigma^4}{p^2}\right],
\end{eqnarray*}
where $\Pi_{D}$ is the orthogonal projection onto $S_D=\text{span}(\phi_\lambda,\lambda\in\Lambda_D)$,  
\[\sigma_\lambda^2=\var(\tilde\varepsilon_{1,\lambda})=\frac{\sigma^2}{p^2}\sum_{h=0}^{p-1}\phi_\lambda^2(t_h),\quad s_\lambda^2=\var(\tilde z_{1,\lambda})=\frac1{p^2}\sum_{h,h'=0}^{p-1}K(t_h,t_{h'})\phi_\lambda(t_h)\phi_\lambda(t_{h'})\]
and 
\[
A_p^{(K)}(\phi,D)=\sum_{\lambda,\lambda'\in\Lambda_D}\left\{\frac1{p^2}\sum_{h,h'=0}^{p-1}K(t_h,t_{h'})\phi_\lambda(t_h)\phi_{\lambda'}(t_{h'})-\int_0^1\int_0^1K(s,t)\phi_{\lambda}(s)\phi_{\lambda'}(t)dsdt\right\}^2,
\]
\[
A_p^{(\sigma)}(\phi,D)=\frac{\sigma^4}{p^2}\sum_{\lambda,\lambda'\in\Lambda_D}\left\{\frac1p\sum_{h=0}^{p-1}\phi_{\lambda}(t_h)\phi_{\lambda'}(t_{h})-\1_{\{\lambda=\lambda'\}}\right\}^2.
\]
\end{theorem}

The first term of the upper bound is a bias term corresponding to the projection step, that decreases with $D$, the dimension of the approximation space. The second term is a variance term that increases with $D$ but contrary to what happens generally in non-parametric statistics, it is bounded by $n^{-1}$ up to a constant under mild assumptions on the orthonormal system (details in Section~\ref{sec:upperspecific}). Indeed, heuristically, when $p$ grows, the term 
$\sigma_\lambda^2$ is of order $\sigma^2/p$ (the variance due to the noise is tempered by the repetition of the observations) and the term $s_\lambda^2$ is of order $\iint K(s,t)\phi_\lambda(s)\phi_\lambda(t)dsdt=E(\langle Z,\phi_\lambda\rangle^2)$ so $\sum_{\lambda\in\Lambda_D} s_\lambda^2$ is bounded by a constant (independent of $D$) of order $E(\|Z\|^2)<+\infty$. By taking $D=\text{card}(\Lambda_D)\leq p$, the second term is of order $n^{-1}$. The third and fourth terms are linked to the discretization and are usually negligible with respect to both the bias and variance terms (see Section~\ref{sec:upperspecific}). The term $\sigma^4/p^2$ is also negligible. 
\textcolor{black}{
\begin{remark}\label{rem1:Gaussian}
We assume that $\{\varepsilon_{i,h}\}_{i=1,\hdots,n; h=0,\hdots,p-1}$ is a sequence of Gaussian errors. But  the result of Theorem~\ref{theo:GeneralEsp} still holds if the vectors $\{\varepsilon_{i,h}\}_{h=0,\hdots,p-1}$ only satisfy Assumption~\ref{Ass:4} with $C_1$ a constant not depending on $i$ (see the proof of Lemma~\ref{Moment}).
\end{remark}
}

We can refine the previous result and obtain similar upper bounds in probability. To state them, we first recall the definition of sub-Gaussian variables. We refer to Section~2.5 of \cite{Vershynin2018} %\cite{Vershynin2012}
for more details.
\begin{definition}
We say that a random variable $W$ is \emph{sub-Gaussian} if 
\[\|W\|_{\psi_2}:=\sup_{q\geq 1}\left\{q^{-1/2}\Big\{E(|W|^q)\Big\}^{1/q}\right\}<\infty.\]
In this case, $\|W\|_{\psi_2}$ is called the \emph{sub-Gaussian norm of $W$}.\\
%for some constant $c$, it satisfies for any $q\geq 1$,
%\begin{equation}\label{SG}
%\Big(\E[|W|^q]\Big)^{1/q}\leq c\sqrt{q}.
%\end{equation}
%The sub-Gaussian norm of $W$, denoted $\|W\|_{\psi_2}$, is defined to be the smallest constant $c$ such that \eqref{SG} is satisfied:
%\[
%\|W\|_{\psi_2}=\sup_{q\geq 1}\left\{q^{-1/2}\Big(\E[|W|^q]\Big)^{1/q}\right\}.
%\]
%A vector $\textbf{W}\in\R^d$ is sub-Gaussian if for any $u\in\R^d$, $u^T\textbf{W}$ is a sub-Gaussian random variable and
%\[
%\|\textbf{W}\|_{\psi_2}:=\sup_{u\in\R^d}\frac{\|u^T\textbf{W}\|_{\psi_2}}{\|u\|_{\ell_2}}.
%\]
\end{definition}
Assumption~\ref{Ass:4} is extended to $p$-dimensional vectors as follows.
\textcolor{black}{\begin{assumption}\label{Ass:q}
 We assume that there exists $C_2>0$ such that 
\begin{equation*}%\label{hyp-SG}
\|v^T{\bf Z}\|_{\psi_2}^2\leq C_2E{(v^T{\bf Z})^2}, \quad v\in\R^p.
\end{equation*}
\end{assumption}}
\textcolor{black}{Assumption~\ref{Ass:q} of Theorem~\ref{theo:Generalproba} is stronger than Assumption~\ref{Ass:4} of Theorem~\ref{theo:GeneralEsp} but it allows us to obtain an inequality in probability, which is stronger than in expectation;} the price to pay is the logarithmic factor in the variance term as show in the following. Using quantities introduced in Theorem~\ref{theo:GeneralEsp}, we then obtain the following result.\\
\begin{theorem}\label{theo:Generalproba}
Let $d$ be fixed.  Then, under Assumption~\ref{Ass:q}, for all $\gamma>0$, with probability larger than $1-2\exp(-1/64\min(\gamma^2,16\gamma\sqrt{n})),$
\begin{eqnarray*}
\|\widehat{\eta}_{\phi,d}-\eta_{\pm,d}^*\|^2&\leq& 5b_d \left[\vvvert\Pi_D\Gamma\Pi_D-\Gamma\vvvert^2+\frac{ (e^{1/2}+\gamma)^2\bar C^2(C_2+1)^2}{n}\left\{\sum_{\lambda\in\Lambda_D}\Big(\sigma_\lambda^2+s_\lambda^2\Big)\right\}^2\right.\\
   &&\hspace{5.5cm} \left.+A_p^{(K)}(\phi,D)+A_p^{(\sigma)}(\phi,D)+\frac{\sigma^4}{p^2}\right],
\end{eqnarray*}
where $\bar C$ is an absolute constant.
\end{theorem}

Observe that if we take $\gamma=8(\beta\log n)^{1/2}$, then for $n$ large enough, the upper bound holds with probability larger than $1-2n^{-\beta}$. In this case, the order of the variance term is the same as for Theorem~\ref{theo:GeneralEsp} up to the $\log n$-factor. 

Theorem~\ref{theo:Generalproba} is based on Assumption~\ref{Ass:q}, namely a control of the sub-Gaussian norm of $v^T{\bf Z}$ for all vectors $v$. Such controls are standard to obtain concentration inequalities which are at the core of the proof of Theorem~\ref{theo:Generalproba}; see for instance  %\cite{Vershynin2012}, 
\cite{Vershynin2018}, \cite{KL2017} or Section~2.3 of \cite{BLM-book}. This assumption enables us to apply large deviation bounds for martingale differences established by \cite{JN} to specific covariance matrices. See Proposition~\ref{concentration} in the Appendix for more details. Observe that Assumption~\ref{Ass:q} is satisfied if ${\bf Z}$ is Gaussian and in this case $C_2$ is an absolute constant %(see Example 5.8 of \cite{Vershynin2012}). 
(see Example 2.5.8 of \cite{Vershynin2018}). 
In a more general setting, we have the following lemma proved in Appendix~\ref{Appendix:A2}.
\textcolor{black}{
\begin{lemma}\label{lemma:A2}
We consider the Karhunen-Lo\`eve representation of $Z$ given by \eqref{eq:KL} and assume that the scores $\zeta^*_d$' are independent and that there exists $M<\infty$ such that
\begin{equation}\label{M}
\sup_{d\in\mathbb N^*}\|\zeta^*_d\|_{\psi_2}\leq M.
\end{equation}
Then Assumption~\ref{Ass:q} is satisfied with $C_2=\kappa M^2$, where $\kappa$ is an absolute constant.
\end{lemma}
}
Under the assumptions of Lemma~\ref{lemma:A2} (and others), \cite{Mas} obtain rates of convergence for the eigenvectors and eigenprojectors in expectation in the case where the curves are observed at all points. In the context of functional Principal Components Regression, these assumptions are classical, we refer e.g. to \cite{hall_theory_2009,crambes_asymptotics_2013}.
%\begin{remark}
%The classical functional assumption (see e.g. \citealt{Mas}) amounts to bound the moments of the variable $\zeta_d^*$ appearing in the Karhunen-Lo\`eve decomposition~\eqref{eq:KL} of $Z$ as follows:
%\[
%\sup_{q\geq 1}\sup_{d\in\mathbb N^*}E\big(|\zeta_d^*|^{2q}\big)\leq q! b^{q-1}
%\]
%for $b>0$ a constant. This type of assumption is not well adapted to the case where the data are discretized but shows strong similarities with ours since $\zeta_d^*=\langle Z,\eta_d^*\rangle \times V(\langle Z,\eta_d^*\rangle)^{-1/2}$. 
%\end{remark}

\textcolor{black}{We have the analog of Remark~\ref{rem1:Gaussian}.
\begin{remark}\label{rem2:Gaussian}
We assume that $\{\varepsilon_{i,h}\}_{i=1,\hdots,n; h=0,\hdots,p-1}$ is a sequence of Gaussian errors. But  the result of Theorem~\ref{theo:Generalproba} still holds if the vectors $\{\varepsilon_{i,h}\}_{h=0,\hdots,p-1}$ only satisfy Assumption~\ref{Ass:q} with $C_2$ a constant not depending on $i$ (see the proofs of Lemmas~\ref{SG1} and \ref{SG2}).
\end{remark}
}
\begin{remark}\label{rem:valeurspropres}
The first step of the proof of our results consists in applying Bosq inequalities to bound 
$\|\widehat{\eta}_{\phi,d}-\eta_{\pm,d}^*\|$. Similar bounds also hold for $|\widehat\mu_d-\mu_d^*|$. See Section \ref{proof:theo2-3} of the Appendix for more details. Therefore, bounds of the previous theorems also hold for $|\widehat\mu_d-\mu_d^*|^2$. Obtaining lower bounds for the estimation of the eigenvalues remains an open interesting question. 
\end{remark}
%%%%%%%%%%%%%%%%%%%%%%%%%%%%%%%%
\subsection{Upper bound for histograms}\label{sec:upperspecific}
In this paragraph, we specify our results for the case of histograms. The histogram system is defined as follows (see Section~7.3 of \cite{Massart} for instance).
\begin{definition}\label{histo:def}
Let $\Lambda_D=\{0,\ldots,D-1\}.$ For any $\lambda\in\Lambda_D$,
\begin{equation*}
\phi_\lambda(t)=D^{1/2}\times\1_{I_\lambda}(t),\quad t \in [0,1], 
\end{equation*}
with $I_\lambda=(\lambda/D,(\lambda+1)/D]$. For any $(\lambda,\lambda')\in\Lambda^2$, $\langle \phi_\lambda, \phi_{\lambda'}\rangle=\1_{\{\lambda=\lambda'\}}$. 
\end{definition}
In the sequel, we consider the following assumption.
\textcolor{black}{\begin{assumption}\label{Ass:histo}
The integer $D$ is such that $D$ divides $p$.
\end{assumption}}
In this framework, all terms appearing in upper bounds of Theorems~\ref{theo:GeneralEsp} and \ref{theo:Generalproba} can be easily controlled.
\begin{proposition}\label{cor:boundH}
Under Assumption~\ref{Ass:histo}, if $P_Z\in  \mathcal R_\alpha(L)$, we have
\[
\vvvert\Pi_D\Gamma\Pi_D-\Gamma\vvvert^2\leq \frac{16L\|K\|_\infty}{(\alpha+1)^2} D^{-2\alpha}, 
\] 
\begin{eqnarray*}
A_p^{(K)}(\phi,D)&\leq& \frac{16\|K\|_\infty L}{(\alpha+1)^2}p^{-2\alpha},\quad A_p^{(\sigma)}(\phi,D)=0
\end{eqnarray*}
and
\[
\sum_{\lambda\in\Lambda_D}\Big(\sigma_\lambda^2+s_\lambda^2\Big)\leq\|K\|_\infty+ \frac{\sigma^2 D}{p}.
\]
\end{proposition}
Combining Proposition~\ref{cor:boundH} with Theorems~\ref{theo:GeneralEsp} and \ref{theo:Generalproba}, we finally deduce the following corollary.
\begin{corollary}\label{cor:final}
Let $d$ be fixed. Assume that $P_Z\in  \mathcal R_\alpha(L)$ and \textcolor{black}{$D=~p$}.
Under Assumption~\ref{Ass:4},
 \begin{eqnarray*}
E\Big(\|\widehat{\eta}_{\phi,d}-\eta_{\pm,d}^*\|^2\Big)&\leq& b_d  \left\{\frac{B(L,K,\alpha)}{p^{2\alpha}}+\frac{5\sigma^4}{p^2}+\frac{V_1(K,\sigma, C_1)}{n}\right\}
\end{eqnarray*}
and under Assumption~\ref{Ass:q}, for any $\beta>0$, for $n$ large enough,  with probability larger than $1-2n^{-\beta}$,
 \begin{eqnarray*}
\|\widehat{\eta}_{\phi,d}-\eta_{\pm,d}^*\|^2&\leq& b_d  \left\{\frac{B(L,K,\alpha)}{p^{2\alpha}}+\frac{5\sigma^4}{p^2}+\frac{V_2(K,\sigma, C_2,\beta)\log n}{n}\right\},
 \end{eqnarray*}
 where $B(L,K,\alpha)$ depends on $L$, $\|K\|_\infty$ and $\alpha$ and $V_1(K,\sigma, C_1)$ (resp. $V_2(K,\sigma, C_2,\beta)$) depends on $\|K\|_\infty$, $\sigma$ and $C_1$ (resp. $\|K\|_\infty$, $\sigma$, $C_2$ and $\beta$) (the constants $B(L,K,\alpha)$, $V_1(K,\sigma, C_1)$ and $V_2(K,\sigma, C_2,\beta)$ are deterministic).
\end{corollary}
Since $\alpha\leq 1$, the term $5\sigma^4/p^2$ is not larger than the first term $B(L,K,\alpha)/p^{2\alpha}$ (up to a constant), and the noise of the observations has no influence on the rates (as soon as the noise level is a constant). In particular, under Assumption~\ref{Ass:4},
 \[
 \sup_{P_Z\in\mathcal R_\alpha(L)}E\Big(\|\widehat{\eta}_{\phi,d}-\eta_{\pm,d}^*\|^2\Big)\leq C\Big(p^{-2\alpha}+n^{-1}\Big),
 \]
 for $C$ a constant. This upper bound and the lower bound of Theorem~\ref{thm:borne_inf_fp} match, meaning that our estimation procedure is optimal in our setting. Observe that Assumption~\ref{Ass:4} is very mild. If we replace it with the stronger Assumption~\ref{Ass:q}, we obtain a control in probability, coming from  exponential bounds on probabilities of large deviations (for the Frobenius norm) for specific matrices. The price to pay is a logarithmic term in the variance term. 
 
As expected, parameters $p$ and $n$ have a strong influence on rates. In our framework with two asymptotics very different in nature, we note that if $p$ is large enough (depending on $n$ and $\alpha$), then our procedure achieves the parametric rate $n^{-1}$ already obtained by \citet{DPR82} and \cite{Mas} when the curves $(Z_1, \hdots,Z_n)$ are fully observed and without noise. It means that discretization has no influence on theoretical performances. Conversely, if $p$ is not very large with respect to $n$, discretization  has a deep impact and rates depend strongly on the underlying smoothness of the curves observed in a noisy setting. The obtained rate $p^{-2\alpha}+n^{-1}$ describes very precisely the competition between the number of discretization points and the number of observations in functional principal component analysis. To the best of our knowledge, these rates are new.

Finally, let us emphasize the simplicity of our optimal estimation procedure. It is based on the most classical ideas: projection by using piecewise constant bases and empirical mean estimation. In particular, regularization is not necessary and the knowledge of $\alpha$ is not required. The use of such standard tools may be surprising in view of results obtained by \cite{Johnstone01} and \cite{BS06}. 
But, as already mentioned, functional principal component analysis is a very specific setting. The rates we obtain have the same shape to those of \citet{DP19} for which strong assumptions on the covariance operator (finite rank and analyticity of the eigenvalues $\eta_d^*$) are required and the noise may exhibit local correlations. 
 \citet{BX15} obtained a rate of convergence that is difficult to compare with ours, because their assumptions, concerning the rate of decay of the eigenvalues, differ significantly from our regularity assumption on the process. In the case of Brownian motion the rate of convergence of the $\mathbb L^2$- risk of reconstruction of the estimator of \citet{BX15} is of order $\log^2(n)n^{-1}+p^{-1/3}$ which is suboptimal compared to the minimax rate of $n^{-1}+p^{-1}$ that we have proven.  In a different statistical model where observational times are random, \citet{Hall} obtained an $\mathbb L^2$ convergence rate $n^{-2r/(2r+1)}$ for kernel estimators when $\eta_d^*$ has a $r$-th order derivative even if the number of observations per curve is bounded by a constant. This clearly shows the impact of the nature of the observations. The randomness of observational times may allow to circumvent the sparse sampling scheme for individuals and consistency may be achievable, which is not the case for our statistical model with deterministic equispaced observational times.
  
 %\vinc{Revenir ici sur les resultats de Bunea et Xiao (2015), Descary \& Panaretos (2016) et Hall et al. (2006). Peut-on dire que, en particulier pour Bunea et Xiao (2015), l'aspect fonctionnel n'est pas pris en compte convenablement ?}
%\textcolor{black}{Relier a ces resultats a ceux de Bnuea et Xiao (2010). Ils sont sous-optimaux, notamment pour la cas du mouvement brownien, car fPCA traite sous l'angle de l'estimation matrice (matrice de covariance).}

%%%%%%%%%%%%%%%%%%%%%%%%%%%%%%%%
%%%%%%%%%%%%%%%%%%%%%%%%%%%%%%%%
\section{Simulation results}\label{sec:simus}
We assess the statistical performance of functional principal components estimators with simulations. We consider two eigenfunctions such that $\eta_1^*(\cdot)=\sqrt{2}\sin(2\pi\cdot)$ and $\eta^*_2(\cdot)=\sqrt{2}\cos(2\pi\cdot)$, with eigenvalues $\mu_1^*=1.1$ and $\mu_2^*=0.1$. Simulated functional data are sampled on regularly spaced discretization points $t_h=h/(p-1)$ with $h=0,\ldots,p-1$, and we compute the covariance matrix $\Sigma$ such that:
$$\Sigma_{h,h'}=\sum_{d=1}^2\mu_d^*\eta_d^*(t_h)\eta_d^*(t_{h'})+\sigma^21_{h=h'}$$ from which we sample $n$ random functions 
$Y_1,\hdots,Y_n\sim\mathcal{N}(0,\Sigma)$, following Model \eqref{model}. Then we consider different values for the number of observations $n\in\{256,512,1024,2048,4096\}$ to study the asymptotic performance of our estimators, and we will also consider different values of $p\in\{16,32,64,128,256\}$ to study the impact of discretization. The noise level $\sigma$ is chosen to match a given signal to noise ratio defined by $\sigma^{-2}\sum_{d=1}^2\mu_d^*$ (the variance of the signal divided by the variance of the noise), that takes value in $\{0.25,1,4\}$. We consider two smoothing systems, histograms and the Haar wavelets, as detailed in the Appendix. We report average the performance on $nb_{test}=100$ independent simulations. Even if our theoretical results do not include regularized estimators, we also consider a hard thresholded version of these estimators to improve reconstruction (as detailed in the Appendix). 
%%%%%%%%%%%%%%%%%
\subsection{Reconstruction Errors}
To assess the empirical performance of our approach, we study the behavior of the mean reconstruction error
$$E\Big( \|\eta^*_{\pm,d}-\widehat{\eta}_{\phi,d}\|^2\Big)$$
 according to the number of observations $n$, the size of the discretization grid $p$, and the signal to noise ratio. More precisely, we introduce a second finer grid $s_{h}=h/p'$, with $ h = 0,\hdots,p'-1$, such that $p' \gg p$ ($p'=8192$ in practice) and use the approximation
$$ 
\mathbb{E} \Big( \|\eta^*_{\pm,d}-\widehat{{\eta}}_{\phi,d}\|^2\Big) \approx \frac{1}{nb_{test}} \sum_{j=1}^{nb_{test}} 
 \frac{1}{p'}\sum_{h=1}^{p'} \Big\{ \eta^*_{\pm,d}(s_h)-\widehat{\eta}^{j}_{\phi,d}(s_h) \Big\}^2.
$$
To deduce the values of our estimator outside of the original grid, we use the piecewise constant property of the Haar and the histogram systems. In the following we also compute the estimation error on eigenvalues and assess $E\{ (\mu_d^*-\widehat{\mu}_{\phi,d})^2 \}.$

\subsection{Results}

%As expected, the reconstruction errors of the first eigenfunction based on the Haar or the histogram bases are the same, since without regularization (no thresholding) the recovered signal after the projection step is the same as the original data (Figure \ref{fig:error_noncv} ). 

The empirical error rates on eigenfunctions match the theoretical ones (Fig. \ref{fig:error_eta_noncv}), with orders $(\mu^*_1-\mu_2^*)^{-2}(n^{-1}+p^{-2\alpha})$ ($\alpha=1$ in our case) for the first eigenfunction estimator and $(\mu^*_2-\mu_3^*)^{-2}(n^{-1}+p^{-2\alpha})$ for the second (Fig. \ref{fig:error_noncv eta2} and \ref{fig:error_noncv mu2}, $\mu_3^*=0$ in our setting). Computed errors exhibit a double asymptotic behavior in $n$ and $p$. The rates in $n$ are slower than those in $p$, and exactly match $n^{-1}$ and $p^{-2}$. Also, the difference in terms of mean square errors between the first and the second eigenfunctions is due to the gap, since $\mu_1^*-\mu_2^*=10(\mu_2^*-\mu_3^*)$, which means that the estimation of the second eigenfunction is $10$ times harder in terms of speed of convergence. The estimation error on eigenvalues also matches the theoretical upper bound (Fig \ref{fig:error_mu_noncv}). 

%However, note that the estimator of the second eigenvalue does not suffer any aggravation due to the gap. On the other hand, we know that in theory (see \cite{DPR82} Proposition 8) the reconstruction error of the un-regularized estimator should behave roughly as $2\mu_d^2/n$,  which explains why the reconstruction error of the first estimator is so high compared to the second one ($\mu_1^2=121 \mu_2^2$). 

\begin{figure}
	\begin{center}
	\begin{tabular}{cc}
		\includegraphics[scale=0.8]{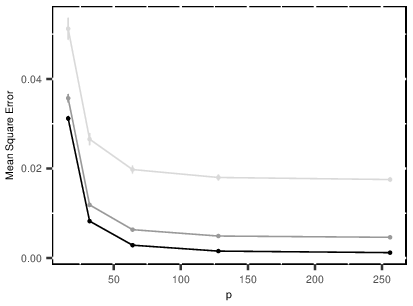} & 
		\includegraphics[scale=0.8]{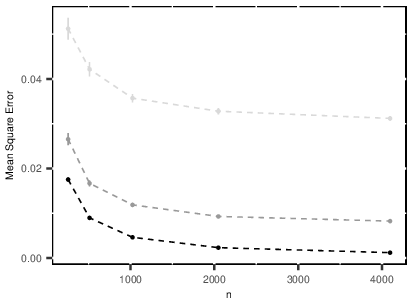}
	\end{tabular}
	\end{center}
	\caption{Mean square error for the first eigenfunction $\eta_1^*$ according to the number of discretization points $p$ (left), and the number of samples $n$ (right). Left: the number of samples is $n \in \{ 256,1024,4096\}$ (light gray, gray, black respectively). Right: the number of discretization points is $p \in \{16,32,256\}$ (light gray, gray, black respectively). The signal to noise ratio is 0.25.\label{fig:error_eta_noncv}}
\end{figure}

\begin{figure}
	\begin{center}
	\begin{tabular}{cc}
	\includegraphics[scale=0.8]{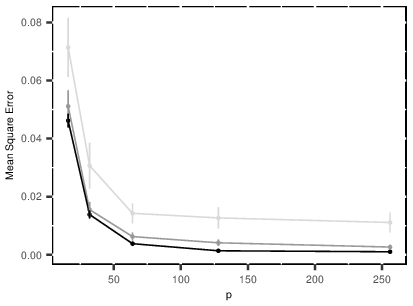} & 
	\includegraphics[scale=0.8]{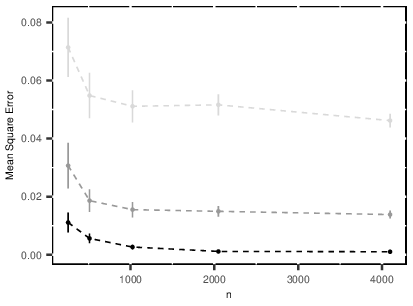}
\end{tabular}
	\end{center}
	\caption{Mean square error for the first eigenvalue $\mu_1^*$ according to the number of discretization points $p$ (left), and the number of samples $n$ (right). Left: the number of samples is $n \in \{ 256,1024,4096\}$ (light gray, gray, black respectively). Right: the number of discretization points is $p \in \{16,32,256\}$ (light gray, gray, black respectively). The signal to noise ratio is 0.25.\label{fig:error_mu_noncv}}	
\end{figure}

We also show that regularization does not necessarily improve the rate of convergence of eigen-elements (Fig. \ref{fig:error_cv_grid_size} and \ref{fig:error_cv_sample_size}). We showed that projection-based functional principal component should attain minimax rates without any regularization. Consequently, at best, regularization should induce a better variance but comparable rates. Since one could operate only on the observed part of the function, at best, one could improve results on the grid without going beyond $n^{-1}$, which is already attained by the non-smoothed estimator. 

%Last, our simulations show that regularization does not necessarily improve the quality of reconstruction (Fig. S3 and S4). Indeed, we showed that projection-based functional principal component should attain minimax rates without any regularization. Consequently, at best,  regularization should induce a better variance but comparable rates (which is not the case here as seen in Figure \ref{fig:variance}  \textcolor{red}{c'est quelle figure ici ?}\textcolor{green}{La figure représente la variance de l'estimateur, elle n'existe plus}). The situation is more obvious when one looks at reconstruction examples (Figure \ref{fig:reconstruction cv}\textcolor{green}{(c'est un exemple de reconstruction après cross validation)}. Since one could operate only on the observed part of the function, at best, one could improve results on the grid without going beyond $n^{-1}$, which is already attained by the non smoothed estimator. Regularization does not bring any improvement to the convergence rates for eigenvalues (Figure \ref{fig:error value cv}\textcolor{green}{erreur valeur propres cv}). 

%Figure \ref{fig:compraison error value}\textcolor{green}{Figure comparant directement cv et non cv vitesse Fig3 et Fig4} compares both approaches in terms of rate of convergence.

% \textcolor{red}{je ne crois pas que le cas infini soit dans les figures}

\section{Applications}
Single-cell genomics has been made possible thanks to technological breakthroughs that have allowed the measurement of gene expression \cite{macosko_highly_2015} at the single-cell resolution. These advances have revolutionized our view of the complexity of living tissues, in their normal or pathological states, and have produced complex high dimensional data. Immunology, and in particular studies of lymphocytes differentiation have focused much attention. Indeed, upon an acute infection, pathogen-specific CD8 T cells are activated, proliferate drastically and differentiate short-term (ou short-lived) effector cells displaying the ability to eliminate infected cells. the response also generates a small population of cells being part of the long-term immune response, the so-called memory cells conferring protection to the host. In \cite{pmid32414833}, the authors collected single antigen-specific CD8 T cells in the spleen of mice after LCMV (lymphocytic choriomeningitis virus) acute infection and measured the expression of genes at time points 0, 3, 4, 5, 6, 7, 10, 14, 21, 32 and 90 days post-infection (dpi). Hence, these data offer the unique opportunity to study the evolution of gene expression throughout an immune response. Many methodological questions can be addressed with these data, and we consider here a short analysis to illustre functional PCA on original data. Average expressions were considered over cells such that the data correspond to the expression of genes over time. Expression data were normalized thanks to the \texttt{SCTransform} procedure of the \texttt{Seurat} package \citep{Seurat} that corrects counts for overdispersion, and provides corrected counts for which the Gaussian approximation is reasonable. Then genes were selected by using the \texttt{FindVariableFeatures} function \citep{Seurat}, resulting in 4851 genes with averaged temporal expression that constitute the input of our model ($Y_i(t)$, $n=4851$, $p=11$). When performing functional PCA on those data using the histogram basis, the rule of thumbs suggests to select 3 functional principal components, and $k$-means clustering exhibits 3 clusters with very distinct average expression profiles (Fig. \ref{fig:scRNASeq}). Interestingly, cluster 1 consists of only two genes Ccl5 and Malat1 that are known to be involved in immune cells activation and recruitement \citep{pmid29559701} and CD8+ differentiation \citep{pmid35593887}, respectively. Their kinetics of expression appears then as a major feature of the gene-expression response to viral infection: they are down-regulated between days 0-8, then up-regulated between 9-90 days, whereas the 31 genes of cluster 2 are up and then down-regulated (Gzmb, Ncl, Tuba1b, Hsp90aa1, Npm1, Hspe1, Ptma, Ran, Hsp90ab1, Actg1, Mif, Tubb5, Eif5a, Ldha, Plac8 and ribosome proteins), the other genes showing a stable average expression stable over time (cluster 3). This first analysis allows us to illustrate the interest of functional PCA on gene expression data, and the underlying gene regulatory networks that structure these gene expression dynamics will deserve a deeper study.

\begin{figure}
	\begin{center}
		\includegraphics{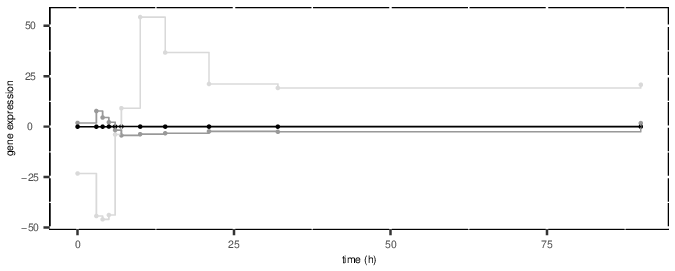}
	\end{center}
	\caption{Average gene expression over time for lymphocytes CD8+ cells following a viral infection, as inferred by functional PCA followed by $k$-means clustering. Curves correspond to the averages of gene expression for the 3 clusters (1-light grey, 2-dark grey, 3-black). \label{fig:scRNASeq}}	
\end{figure}

Genomics offers another original application of functional principal component analysis, to reduce the dimensionality of data that are structured in one dimension along the genome. As an illustration, we consider the fine mapping of replication origins in the human genome, that constitutes the starting points of chromosomes duplication. Replication origins are under a very strong spatio-temporal control, and are part of the integrity maintenance of genomes. The investigation of their spatial organization has become central to better understand genomes architecture and regulation, which remains challenging due to a complex interplay between genetic and epigenetic regulations. Part of the genetic component of their regulation involve particular sequence motifs, called G-quadruplexes, that have the property to form complex four-stranded structures whose role in replication remains unsolved. A crucial aspect to better understand their implication is to determine if these sequence motifs have a preferential positioning upstream replication origins. To investigate this matter, we considered the $\sim$130,000 replication origins of the human genome \citep{PCA14}, and we defined by $Y_i(t)$ the process that equals 1 if there is a G-quadruplex at position $t$ in replication origin $i$, taking motifs coordinates from \cite{ZZH20}. Hence this application goes beyond the Gaussian setting of our model, and shows that extension to count data is also effective with histogram-based functional-PCA. By convention, $t=0$ corresponds to the peak of replication, and we consider positions 500 bases upstream this peak (in negative coordinates). The continuous aspect of the model is not mandatory since positions along the genome are discrete. However, the functional setting allows us to consider the spatial dependencies between the occurrences of these motifs, which is very informative. Given the discrete nature of the data, we smoothed the data using the histogram system, with bin size of 25 bases (corresponding to the average size of G-quadruplexes). Then we performed functional principal component analysis, and we used functional principal components to perform a downstream clustering. We projected every observed curve on principal components to obtain a new representation of the functional data based on general terms $\langle \widetilde{Y}_i, \widehat{\eta}_{\phi,d}\rangle$, and performed a $k$-means clustering to regroup replication origins that share the same spatial distribution of these G-quadruplex motifs. We considered 6 principal components along with 6 clusters, and we considered the spatial distribution of G-quadruplexes within clusters as a result (Figure \ref{fig:application}). Functional principal component analysis appears to catch the spatial structure that makes the clusters, as different clusters of replication origins are characterized by specific patterns of G-quadruplexes accumulation upstream the replication peak. Interestingly, the observed periodicity can be related to a biophysical property of chromatin fibers. Indeed, the DNA molecule is in the form of chromatin fibers in the nucleus, wrapped around the so-called nucleosomes with a periodicity of 144 base pairs. The formation of stable G-quadruplexes has been shown to take place in nucleosome-free regions \citep{PAA19}, hence, the periodicity of their accumulation upstream replication origins indicates that their positioning is directly linked to the epigenetic context of replication initiation. These new biological results are currently under further investigation, and developing a framework for functional-PCA dedicated to count data could be an interesting research direction \citep{BSR20}.

\begin{figure}
	\begin{center}
		\includegraphics{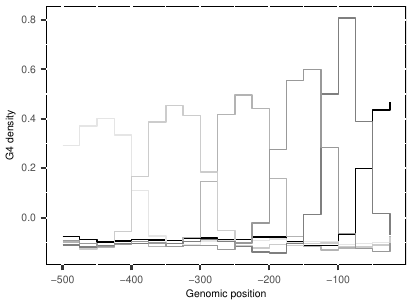}
	\end{center}
	\caption{Density of G-quadruplexes accumulation in human replication origins clusters, determined by functional principal component analysis combined with $k$-means clustering. Each color correspond to a particular cluster.\label{fig:application}}	
\end{figure}
%%%%%%%%%%%%%%%%%%
%%%%%%%%%%%%%%%%%%
\section{Conclusion and perspectives}
In this paper, we have established optimal rates of convergence for estimating the eigenfunctions of the covariance operator of a corrupted process observed on a fixed and regular grid. It is shown that the minimax rate is of order $p^{-2\alpha}+n^{-1}$, revealing the behavior of rates with respect to the parameters $n$, $p$ and $\alpha\in(0,1]$ that design the number of repetitions, the size of the regular grid and the smoothness of the process respectively. In this framework with a double asymptotic in $n$ and $p$, the phase transition occurs when $p$ is of order $n^{1/(2\alpha)}$: 
\begin{itemize}
\item For $p$ larger than $n^{1/(2\alpha)}$, the obtained rate is $n^{-1}$ and the problem has intrinsically the same difficulty if the $n$ curves would be available entirely.
\item For $p$ smaller $n^{1/(2\alpha)}$, the rate is $p^{-2\alpha}$ regardless the number of observed curves. \end{itemize} In particular, $p\to+\infty$ is required to achieve consistency.
These rates are new and are fundamentally different from classical nonparametric rates of the form $(np)^{-2\alpha/(2\alpha+1)}$ in a setting with $n\times p$ observations, revealing main differences with the framework where the sampling scheme is random \citep{Hall} and typical of longitudinal data analysis. In our setting, optimality is achieved by diagonalizing the empirical operator after simple projections by using histograms with $p$ bins. Surprisingly, it means that smoothing and the knowledge of $\alpha$ are not required. These results have important theoretical and practical implications in many aspects regarding the sampling scheme, the used methodology or the grid size if it is left to the practitioner (the choice of the grid itself is a research field on its own and out of the scope of the paper, we refer to \cite{MG96,Seleznjev00} for articles on optimal grid selection). These conclusions bear some similarities with those drawn by \cite{CY11} for mean function estimation. 

Following this work, several perspectives can be envisaged. The first one concerns the range of $\alpha$ and it should be interesting to extend our theoretical results to any smoothness $\alpha\in\R_+^*$, even if it means changing the definition of the class $\mathcal R_\alpha(L)$. This extended framework with higher smoothness raises several open questions: Are rates different? If yes, do we need to introduce smoothing? Can we still consider histograms or do we have to project on smoother bases such as wavelets, splines or Fourier bases?

Our minimax rates depend on the sequence of eigenvalues $(\mu^*_d)_d$ through the terms $(b_d)_d$'s that are viewed as (unknown) constants in our paper. It would be very interesting to investigate how such parameters influence rates when they become very small. Concerning the problem of estimating eigenvalues, Remark~\ref{rem:valeurspropres} provides some upper bounds but optimality of these upper bounds remain a very challenging problem, out of the scope of this paper. 

\textcolor{black}{Besides, as pointed out by \cite{HJ20a}, assuming that the noise is i.i.d. may be unrealistic in some specific situations. \cite{HJ20a} assume for instance that the vector of errors $(\varepsilon_{i,h})_{h=0,\hdots,p}$ is a stationary  Gaussian process, while \cite{DP19} model the errors as realizations of a stochastic process with short term dependency. Both articles obtain convergence rates for the eigenfunctions in the case where the covariance operator $\Gamma$ is of finite-rank. Based on these works, it could be of interest, in a future work, to extend the minimax approach developed in this paper to models allowing non i.i.d. errors.}

%%%%%%%%%%%%%%%%%%%%%
%%%%%%%%%%%%%%%%%%%%%
\bibliographystyle{abbrvnat}
\bibliography{biblio}

\section*{Acknowledgements}

The authors want to thank referees for their helpful comments which have improved the presentation of this work. We also would like to thank Christophe Arpin and Margaux Preux for their help on the scRNAseq data analysis. The research was supported by a grant from R\'egion Ile-de-France, and by grants from the Agence Nationale de la Recherche  ANR-18-CE45-0023 SingleStatOmics and ANR-18-CE40-0014 Smiles.

%%%%%%%%%%%%%%%%%%%%%
%%%%%%%%%%%%%%%%%%%%%
%%%%%%%%%%%%%%%%%%%%%
%%%%%%%%%%%%%%%%%%%%%

\appendix

\section{Simulation study} \label{appendix:simulation}

We consider two smoothing systems, the histogram system, and the Haar wavelet system. In the case of histograms, we denote by $D$ the number of bins (such that $D$ divides $p$ in practice), then $\Lambda_D=\{0,\ldots,D-1\}$ and
$$
	\phi_\lambda(t)= D^{1/2}\times\1_{(\lambda/D,(\lambda+1)/D]}(t),\quad t \in [0,1], \lambda \in \Lambda_D.
$$
Then in the case of the Haar system, we consider $(\varphi_{0,0}, \psi_{j,k}, j=0,\hdots,J, k=0,\hdots, 2^j-1)$, with $J+1 = \log_2(p)$, $\varphi_{0,0}$ the scaling function of a multi-resolution analysis (father wavelet) and $\psi_{j,k}$ the associated mother wavelets, such that 
\begin{align*}
	\varphi_{0,0}(t)&=\1_{[0,1]}(t),\\
	\psi_{j,k}(x)&=2^{j/2}\1_{\big[ \frac{2k-2}{2^{j+1}}; \frac{2k-1}{2^{j+1}} \big]}(t) - 2^{j/2}\1_{\big( \frac{2k-1}{2^{j+1}}; \frac{2k}{2^{j+1}} \big]}(t), \quad t \in [0,1].
\end{align*}
We introduce a cross-validation procedure to regularize the eigenfunctions estimators. For each fold $r\in\{1,\ldots,n_{\text{folds}}\}$ we split the observations $Y=(Y_1,\hdots, Y_n)$ into two training and test sets $Y^{\text{train}_r}, Y^{\text{test}_r}$ such that $ \mid \text{train}_r \cup \text{test}_r \mid= n$. Then we introduce $\zeta$, a thresholding parameter, and we set
$$
\widehat{Y}_{i,\zeta}(t) =\widetilde{y}_{i,0,0}\varphi_{0,0}(t) +\sum_{j=0}^{J}\sum_{k=0}^{2^j-1} \widetilde{y}_{i,j,k} \1_{|\widetilde{y}_{i,j,k}|>\zeta}\psi_{j,k}(t),\quad i=1,\ldots,n,\quad t\in[0,1],$$
with 
\begin{align*}
\widetilde y_{i,0,0}=\frac1 p\sum_{h=0}^{p-1}Y_i(t_h)\varphi_{0,0}(t_h), \\
\widetilde y_{i,j,k}=\frac1 p\sum_{h=0}^{p-1}Y_i(t_h)\psi_{j,k}(t_h).
\end{align*}
Then we compute the $\widehat{\eta}_{d,\zeta}^r$'s on $Y^{\text{train}_r}$ for each fold such that
$$
(\widehat{\eta}_{d,\zeta}^r)_d \in \argmin_{\langle f_d,f_{d'}\rangle=1_{d=d'}}\sum_{i\in \text{train}_r}\|\widehat{Y}_{i,\zeta}-\sum_{d=1}^2\langle\widehat{Y}_{i,\zeta},f_d\rangle f_d\|^2.
$$
We select $\widehat{\zeta}$, the minimizer of the cross validated errors :
$$\frac{1}{n_{\text{folds}}}\sum_{r=1}^{n_{\text{folds}}}\sum_{i\in \text{test}_r} \|Y_{i}-\sum_{d=1}^2 \langle Y_{i}, \widehat{\eta}_{d,\zeta}^r \rangle \widehat{\eta}_{d,\zeta}^r \|^2.$$ 
Once $\widehat{\zeta}$ is chosen we compute the final estimator $\widehat{\eta}_{d,\widehat{\zeta}}$ as :
$$(\widehat{\eta}_{d,\widehat{\zeta}})_d\in\argmin_{\langle f_d,f_{d'}\rangle=1_{d=d'}}\sum_{i=1}^n\|\widehat{Y}_{i,\widehat{\zeta}}-\sum_{d=1}^2\langle\widehat{Y}_{i,\widehat{\zeta}},f_d\rangle f_d\|^2.$$
We use the same score function to select the number of bins for histograms.

\begin{figure}
	\begin{center}
	\begin{tabular}{cc}
	\includegraphics[scale=0.8]{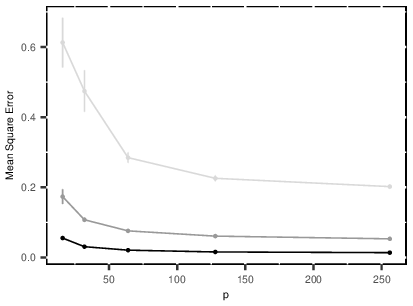} & 
	\includegraphics[scale=0.8]{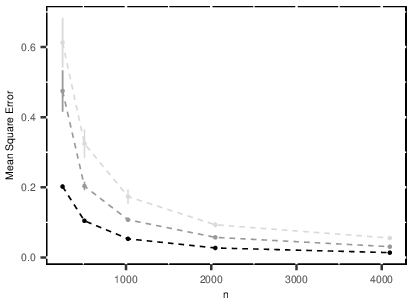}
\end{tabular}

	\end{center}
	\caption{Mean square error for the second eigenfunction $\eta_2^*$ according to the number of discretization points $p$ (left), and the number of samples $n$ (right). Left: the number of samples is $n \in \{ 256,1024,4096\}$ (light gray, gray, black respectively). Right: the number of discretization points is $p \in \{16,32,256\}$ (light gray, gray, black respectively). The signal to noise ratio is 0.25. \label{fig:error_noncv eta2}}
\end{figure}

\begin{figure}
	\begin{center}
	\begin{tabular}{cc}
	\includegraphics[scale=0.8]{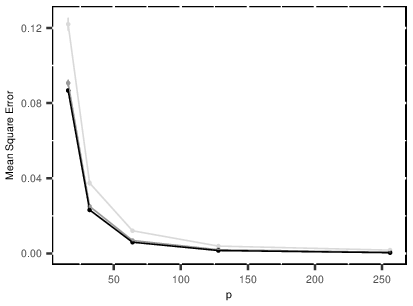} & 
	\includegraphics[scale=0.8]{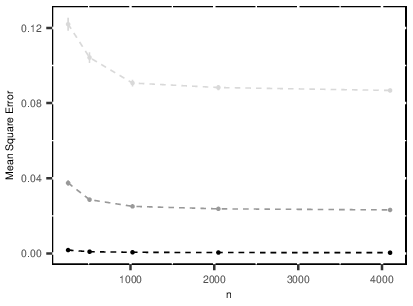}
\end{tabular}

	\end{center}
	\caption{Mean square error for the second eigenvalue $\mu_2^*$ according to the number of discretization points $p$ (left), and the number of samples $n$ (right). Left: the number of samples is $n \in \{ 256,1024,4096\}$ (light gray, gray, black respectively). Right: the number of discretization points is $p \in \{16,32,256\}$ (light gray, gray, black respectively). The signal to noise ratio is 0.25. \label{fig:error_noncv mu2}}
\end{figure}

\begin{figure}
	\begin{center}
		\includegraphics{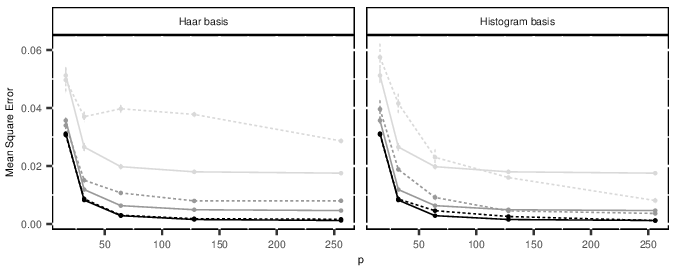}
	\end{center}
	\caption{Mean square error for the first eigenfunction $\eta_1^*$  according to the number of discretization points $p$ and the smoothing system (Haar, left panel, histograms, right panel). The number of samples is $n \in \{ 256,1024,4096\}$ (light gray, gray, black respectively). Dashed line: regularized estimator based on cross validation, plain line: non regularized estimator.  The signal to noise ratio is 0.25. Regularization is performed by cross-validation to choose the number of wavelet coefficients or bins for histograms, as detailed in Appendix \ref{appendix:simulation}}. \label{fig:error_cv_grid_size}
\end{figure}

\begin{figure}
	\begin{center}
		\includegraphics{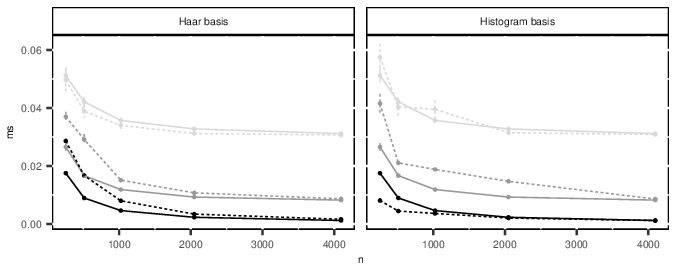}
	\end{center}
	\caption{Mean square error for the first eigenfunction $\eta_1^*$ according to the number samples $n$ and the smoothing system (Haar, left panel, histograms, right panel). The number of discretization points is $p \in \{16,32,256\}$ (light gray, gray, black respectively). Dashed line: regularized estimator based on cross validation, plain line: non regularized estimator. The signal to noise ratio is 0.25. Regularization is performed by cross-validation to choose the number of wavelet coefficients or bins for histograms, as detailed in Appendix \ref{appendix:simulation}}. \label{fig:error_cv_sample_size}.
\end{figure}

\newpage	
%%%%%%%%%%%%%%%%%%%%%%%%%%
%%%%%%%%%%%%%%%%%%%%%%%%%%
\section{Proof of Theorem~\ref{thm:borne_inf_fp}}
\label{proof:A}
To establish Theorem~\ref{thm:borne_inf_fp}, we prove following Propositions~\ref{prop:borne_inf_fpn} and~\ref{prop:borne_inf_fpp}.
%\textcolor{red}{Dans les 2 propositions qui suivent, la v.p. doit etre inferieure a $L$ modulo constante.}
\begin{proposition}
	\label{prop:borne_inf_fpn}
Assume that the rank of the operator $\Gamma$ is larger than 2 and $p\geq 4$. Then, %\textcolor{orange}{for all $L\geq 
\[
\inf_{\widehat\eta_1}\sup_{P_Z\in\mathcal R_\alpha(L)}\mathbb E[\|\widehat\eta_1-\eta_1^*\|^2]\geq c_1n^{-1},\quad \inf_{\widehat\eta_2}\sup_{P_Z\in\mathcal R_\alpha(L)}\mathbb E[\|\widehat\eta_2-\eta_2^*\|^2]\geq c_1n^{-1},
\] 
where $c_1>0$ is a constant depending on $L$, $\alpha$ and $\sigma$. 
\end{proposition}

\begin{proposition}
	\label{prop:borne_inf_fpp}
	There exists a universal constant $c_2>0$ such that
	\[
	\inf_{\widehat\eta_1}\sup_{P_Z\in\mathcal R_\alpha(L)}\mathbb E[\|\widehat\eta_1-\eta_1^*\|^2]\geq c_2p^{-2\alpha}.
	\] 
\end{proposition}
The result of Theorem~\ref{thm:borne_inf_fp} is deduced from Propositions~\ref{prop:borne_inf_fpn} and ~\ref{prop:borne_inf_fpp}, by taking \[c=\frac{1}{2}\min(c_1; c_2)>~0.\]
Studying the case $p\in\{1,2,3\}$ separately, which is easy, provides the result of Theorem~\ref{thm:borne_inf_fp} for any $n\geq 1$ and any $p\geq 1$.

%\textcolor{black}{We remark that in proofs of Propositions~\ref{prop:borne_inf_fpn} and~\ref{prop:borne_inf_fpp}, the largest eigenvalue $\mu_1^*$ (that corresponds to the spectral radius of the covariance operator $\Gamma$) is larger, up to multiplicative constants depending on the eigenfunctions, than the quantity $L$ appearing in the regularity class we consider. This can be related with value of $L$ obtained under $\alpha$-höldérian assumptions on the eigenfunctions detailed in Subsection~\ref{sec:regularity} (see Remark~\ref{rk:L}, p.~\pageref{rk:L}). }
%%%%%%%%%%%%%%%%%%%%%
\subsection{Proof of Proposition~\ref{prop:borne_inf_fpn}} 
\textbf{Proof:} %\textcolor{red}{Without loss of generality, we assume that $\mu^*_3=0.$ En fait, on n'a pas le droit de choisir $\mu^*_1$ et  $\mu^*_2$...} 
Denoting by 
\[
\mathcal R_\alpha^{\mathcal N}(L) = \left\{P_Z\in\mathcal R_\alpha(L)\text{ and }Z\text{ is a Gaussian process}\right\},
\]
we remark that $\mathcal R_\alpha^{\mathcal N}(L)\subset\mathcal R_\alpha(L)$, then, for $d=1,2$,
\[
\inf_{\widehat\eta_d}\sup_{P_Z\in\mathcal R_\alpha(L)}\mathbb E[\|\widehat\eta_d-\eta_d^*\|^2]\geq \inf_{\widehat\eta_d}\sup_{P_Z\in\mathcal R_\alpha^{\mathcal N}(L)}\mathbb E[\|\widehat\eta_d-\eta_d^*\|^2]. 
\] 
Hence, we can restrict the proof of our lower bound to the case of Gaussian processes.

For $t\in[0,1]$, let
\[
a(t)=\sqrt{2}\cos(2\pi t),\quad b(t)=\sqrt{2}\sin(2\pi t).
\]
We observe that
$$\|a\|=\|b\|=1,\quad \langle a,b\rangle=0.$$ 
We then define, for $t\in[0,1]$, 
\[\eta_{A,0}(t)=a(t),\quad \eta_{B,0}(t)=b(t),\]
\[\eta_{A,1}(t)=c_n\Bigg(a(t)+\frac{b(t)}{\sqrt{n}}\Bigg),\quad \eta_{B,1}(t)=c_n\Bigg(b(t)-\frac{a(t)}{\sqrt{n}}\Bigg)\]
and $c_n$ such that 
\[\|\eta_{A,1}\|=\|\eta_{B,1}\|=1\]
and we obtain
$$c_n^2=\Big(1+\frac{1}{n}\Big)^{-1}<1.$$
We then have for $j=0,1$,
$$\|\eta_{A,j}\|=\| \eta_{B,j}\|=1,\quad \langle \eta_{A,j}, \eta_{B,j}\rangle=0.$$
Functions $\eta_{A,j}$ and $\eta_{A,j}$ are also $C^\infty$. Now, we introduce for $j=0,1$,
\[
Z^j(t)=\sqrt{\mu_{A}}\xi_A\eta_{A,j}(t)+\sqrt{\mu_{B}}\xi_B\eta_{B,j}(t),\quad  t\in[0,1],
\]
where $\mu_A$ and $\mu_B$ are two positive constants such that $L/(64\pi^2)\geq \mu_A>\mu_B$ and  $\xi_A\sim\xi_B\sim\mathcal N(0,1)$ with $\xi_A$ and $\xi_B$ independent. 
\begin{remark}
Observe that for any $t\in[0,1]$,
\[Z^0(t)\sim\mathcal N\big(0,\mu_{A}a^2(t)+\mu_{B}b^2(t)\big),\]
\[Z^1(t)\sim\mathcal N\Bigg(0,\mu_{A}c_n^2\Big(a^2(t)+\frac{b^2(t)}{n}+2\frac{a(t)b(t)}{\sqrt{n}}\Big)+\mu_{B}c_n^2\Big(b^2(t)+\frac{a^2(t)}{n}-2\frac{a(t)b(t)}{\sqrt{n}}\Big)\Bigg).\]
\end{remark}
We consider Model~\eqref{model} with $\sigma=0$ such that $Z_1^0, \ldots,Z_n^0$ (resp. $Z_1^1, \ldots,Z_n^1$) are i.i.d copies of $Z^0$ (resp. $Z^1$). It is straightforward to observe that the lower bound established for $\sigma=0$ provides a lower bound for any $\sigma\geq 0$. We then observe $n$ i.i.d.  copies of
\[
\left\{
\begin{array}{c}
Z^0(t_h)=\sqrt{\mu_{A}}\xi_A\eta_{A,0}(t_h)+\sqrt{\mu_{B}}\xi_B\eta_{B,0}(t_h)   \\ 
Z^1(t_h)=\sqrt{\mu_{A}}\xi_A\eta_{A,1}(t_h)+\sqrt{\mu_{B}}\xi_B\eta_{B,1}(t_h)
\end{array}
\right.
\]
for $h=0,\ldots,p-1$. Let, for $j=0,1$, $P_{j}^Z$ the distribution of $Z^j$. We have for any $(t,u)\in [0,1]^2$,
\begin{align*}
	\int_{\C^0}(z(t)-z(u))^2dP_j^Z(z)&=E[(Z^j(t)-Z^j(u))^2]\\
	&=\mu_A\left(\eta_{A,j}(t)-\eta_{A,j}(u)\right)^2+\mu_B\left(\eta_{B,j}(t)-\eta_{B,j}(u)\right)^2\\
	&\leq C\mu_A|t-u|^{2\alpha},
\end{align*}
for $C=64\pi^2$,
since
$$|a(t)-a(u)|^2\leq 8\pi^2|t-u|^2,\quad |b(t)-b(u)|^2\leq 8\pi^2|t-u|^2$$
implies
$$\int_{\C^0}(z(t)-z(u))^2dP_0^Z(z)\leq 16\pi^2\mu_A|t-u|^{2\alpha}$$
$$\int_{\C^0}(z(t)-z(u))^2dP_1^Z(z)\leq 64\pi^2\mu_A|t-u|^{2\alpha}$$

 (we have used that $\mu_A>\mu_B$ and $n\geq 1$ and $|t-u|^{1-2\alpha}\leq 1$, since $\alpha\leq 1$ and the mean value theorem). We  easily deduce that $P_{j}^Z\in\mathcal R_\alpha^{\mathcal N}(L)$ 
since $64\pi^2\mu_A\leq L$. 
This allows to deduce that 
\[
\inf_{\widehat\eta_1}\sup_{P_Z\in\mathcal R_\alpha(L)}E[\|\widehat\eta_1-\eta_1^*\|^2]\geq\inf_{\widehat\eta_1}\sup_{j=0,1} E[\|\widehat\eta_1-\eta_{A,j}\|^2].
\] 
 We obtain similarly
\[
\inf_{\widehat\eta_2}\sup_{P_Z\in\mathcal R_\alpha(L)}E[\|\widehat\eta_2-\eta_2^*\|^2]\geq\inf_{\widehat\eta_2}\sup_{j=0,1} E[\|\widehat\eta_2-\eta_{B,j}\|^2].
\] 
We now prove a lower bound for $E[\|\widehat\eta_1-\eta_{A,j}\|^2]$.
Let $\widehat\eta_1$ an estimator and $\hat\psi$ 
the minimum distance test defined by
\[
\hat\psi= {\arg\min}_{j=0,1} \|\widehat\eta_1-\eta_{A,j}\|^2,
\]
we have for $j=0,1$,
\[
\|\widehat\eta_1-\eta_{A,j}\|\geq\frac12\|\eta_{A,\hat \psi}-\eta_{A,j}\|. 
\]
Now, we have, with $c_n^2=n/(n+1)$,
\begin{eqnarray*}
	\|\eta_{A,\hat \psi}-\eta_{A,j}\|^2&=&\mathbf 1_{\{\hat \psi\neq j\}}\|\eta_{A,0}-\eta_{A,1}\|^2=\mathbf 1_{\{\hat \psi\neq j\}}\left\|a-c_n\Big(a+\frac{b}{\sqrt{n}}\Big)\right\|^2\\	&\geq&\mathbf 1_{\{\hat \psi\neq j\}}\Big((1-c_n)^2+c_n^2/n\Big)\geq  \mathbf 1_{\{\hat \psi\neq j\}}\frac{c_n^2}{n}=\mathbf 1_{\{\hat \psi\neq j\}}\frac{1}{n+1},
\end{eqnarray*}
for any $n$. Then, 
\begin{eqnarray}\label{Borneinf}
	\inf_{\widehat\eta_1}\sup_{P_Z\in\mathcal R_\alpha(L)}\mathbb E[\|\widehat\eta_1-\eta_{1}^*\|^2]&\geq&\frac{1}{4(n+1)}\times\inf_{\hat\psi}\max_{j=0,1}\mathbb P(\widehat\psi\neq j).
\end{eqnarray}
We now prove that the quantity $\inf_{\hat\psi}\max_{j=0,1}\mathbb P(\widehat\psi\neq j)$ can be bounded from below by an absolute positive constant. For this purpose, we control the Kullback divergence between the two models. We have the following lemma proved in Section~\ref{sec:kullback}.
\begin{lemma}\label{lemma:kullback}
Denoting $P_{j}^{obs}$ the distribution of the random vector 
$\mathbf Z^{j,obs}:=(Z^{j}(t_0),\hdots,Z^{j}(t_{p-1})),$
 $KL(P_{1}^{obs},P_{0}^{obs}),$ the Kullback divergence between $P_{1}^{obs}$ and $P_{0}^{obs}$ satisfies, if $p\geq 4$,
\[KL(P_{1}^{obs},P_{0}^{obs})=\frac{1}{2(n+1)}\Big(\frac{\mu_{A}}{\mu_{B}}+\frac{\mu_{B}}{\mu_{A}}-2\Big).\]
\end{lemma}
The result of the lemma entails that $KL((P_{1}^{obs})^{\otimes n},(P_{0}^{obs})^{\otimes n}),$ the Kullback divergence between $(P_{1}^{obs})^{\otimes n}$ and $(P_{0}^{obs})^{\otimes n}$ 
satisfies
\[KL((P_{1}^{obs})^{\otimes n},(P_{0}^{obs})^{\otimes n})=n\times KL(P_{1}^{obs},P_{0}^{obs})=\frac{n}{2(n+1)}\Big(\frac{\mu_{A}}{\mu_{B}}+\frac{\mu_{B}}{\mu_{A}}-2\Big).\]
and then is bounded by a constant $\kappa$ only depending on $\mu_{A}$ and $\mu_{B}$. Therefore, Theorem 2.2 of \citet{tsybakov_introduction_2009} shows that
\[
\inf_{\hat\psi}\max_{j=0,1}\mathbb P(\widehat\psi\neq j)\geq\max\left\{\frac{1}{4}\exp(-\kappa),\frac{1-\sqrt{\kappa/2}}{2}\right\}>0.
\]
and Inequality \eqref{Borneinf} provides the desired lower bound. 
In the same way, we obtain a lower bound for $\inf_{\widehat\eta_2}\sup_{P_Z\in\mathcal R_\alpha(L)}E[\|\widehat\eta_2-\eta_2^*\|^2]$.
%%%%%%%%%
\subsubsection{Proof of Lemma~\ref{lemma:kullback}}\label{sec:kullback}
We first remark that 
\[
\mathbf Z^{j,obs}\sim\mathcal N(0,G_{j}),
\]
where $ G_{j}=([G_{j}]_{k,\ell})_{0\leq k,\ell \leq p-1}$ and
\[
[G_{j}]_{k,\ell}=\mathbb E[Z_{j}(t_k)Z_{j}(t_\ell)]=\mu_{A}\eta_{A,j}(t_k)\eta_{A,j}(t_\ell)+\mu_{B}\eta_{B,j}(t_k)\eta_{B,j}(t_\ell).
\]
Let us explicit $G_0$ and $G_1$. We have
\begin{align*}
[G_{0}]_{k,\ell}&=\mu_{A}\eta_{A,0}(t_k)\eta_{A,0}(t_\ell)+\mu_{B}\eta_{B,0}(t_k)\eta_{B,0}(t_\ell)\\
&=\mu_{A}a(t_k)a(t_\ell)+\mu_{B}b(t_k)b(t_\ell)\\
&=2\mu_{A}\cos(2\pi t_k)\cos(2\pi t_\ell)+2\mu_{B}\sin(2\pi t_k)\sin(2\pi t_\ell)
\end{align*}
and
\begin{align*}
[G_{1}]_{k,\ell}&=\mu_{A}\eta_{A,1}(t_k)\eta_{A,1}(t_\ell)+\mu_{B}\eta_{B,1}(t_k)\eta_{B,1}(t_\ell)\\
&=\mu_{A}c_n^2\Big(a(t_k)+\frac{b(t_k)}{\sqrt{n}}\Big)\Big(a(t_\ell)+\frac{b(t_\ell)}{\sqrt{n}}\Big)+\mu_{B}c_n^2\Big(b(t_k)-\frac{a(t_k)}{\sqrt{n}}\Big)\Big(b(t_\ell)-\frac{a(t_\ell)}{\sqrt{n}}\Big)\\
&=2\mu_{A}c_n^2\Big(\cos(2\pi t_k)+\frac{\sin(2\pi t_k)}{\sqrt{n}}\Big)\Big(\cos(2\pi t_\ell)+\frac{\sin(2\pi t_\ell)}{\sqrt{n}}\Big)\\
&\hspace{2cm}+2\mu_{B}c_n^2\Big(\sin(2\pi t_k)-\frac{\cos(2\pi t_k)}{\sqrt{n}}\Big)\Big(\sin(2\pi t_\ell)-\frac{\cos(2\pi t_\ell)}{\sqrt{n}}\Big).
\end{align*}
We now determine eigenelements of $G_0$ and $G_1$. For this purpose, we use the following lemma proved in Section~\ref{sec:trigo}.
\begin{lemma}\label{lemma:trigo}
Let $p\geq 3$. With $t_k=k/(p-1)$, for $k=0,\ldots,p-1$, we kave
$$\sum_{k=0}^{p-1}\cos(2\pi t_k)=1,\quad \sum_{k=0}^{p-1}\sin(2\pi t_k)=0,\quad \sum_{k=0}^{p-1}\sin(4\pi t_k)=0.$$
Furthermore, for $p\geq 4$,
$$\sum_{k=0}^{p-1}\cos^2(2\pi t_k)=\frac{p+1}{2},\quad \sum_{k=0}^{p-1}\sin^2(2\pi t_k)=\frac{p-1}{2}.$$
\end{lemma}
We set $e=(e_k)_{0\leq k \leq p-1}$ and $f=(f_k)_{0\leq k \leq p-1}$ with
\[e_k=\sqrt{\frac{2}{p+1}}\cos(2\pi t_k)=\frac{a(t_k)}{\sqrt{p+1}},\quad f_k=\sqrt{\frac{2}{p-1}}\sin(2\pi t_k)=\frac{b(t_k)}{\sqrt{p-1}}.\]
Lemma~\ref{lemma:trigo} shows that
\[\|e\|_{\ell_2}=\| f\|_{\ell_2}=1,\quad \langle e,f\rangle_{\ell_2}=0.\]
We then complete $(e,f)$ so that we have an orthonormal basis of $\R^p$, denoted $\mathcal B$. We observe that
\[
[G_{0}]_{k,\ell}=\mu_{A}(p+1)e_ke_l+\mu_{B}(p-1)f_kf_\ell,
\]
which entails
\[G_0e=\mu_{A}(p+1)e,\quad G_0f=\mu_{B}(p-1)f.\]
Similarly,
\begin{align*}
[G_{1}]_{k,\ell}
&=\mu_{A}c_n^2\Bigg(\sqrt{p+1}e_k+\sqrt{\frac{p-1}{n}}f_k\Bigg)\Bigg(\sqrt{p+1}e_\ell+\sqrt{\frac{p-1}{n}}f_\ell\Bigg)\\
&\hspace{2cm}+\mu_{B}c_n^2\Bigg(\sqrt{p-1}f_k-\sqrt{\frac{p+1}{n}}e_k\Bigg)\Bigg(\sqrt{p-1}f_\ell-\sqrt{\frac{p+1}{n}}e_\ell\Bigg),
\end{align*}
which entails
\[G_1e=\Bigg(\mu_{A}c_n^2(p+1)+\mu_{B}c_n^2\frac{p+1}{n}\Bigg)e+\Bigg(\mu_{A}c_n^2\sqrt{\frac{p^2-1}{n}}-\mu_{B}c_n^2\sqrt{\frac{p^2-1}{n}}\Bigg)f\]
and
\[G_1f=\Bigg(\mu_{A}c_n^2\sqrt{\frac{p^2-1}{n}}-\mu_{B}c_n^2\sqrt{\frac{p^2-1}{n}}\Bigg)e+\Bigg(\mu_{A}c_n^2\frac{p-1}{n}+\mu_{B}c_n^2(p-1)\Bigg)f.\]
We have shown that ${\bf Z}^{0,obs}$ and ${\bf Z}^{1,obs}$ are supported by the hyperplane spanned by $e$ and $f$ and the variance-covariance matrices for ${\bf Z}^{0,obs}$ and ${\bf Z}^{1,obs}$ expressed on $(e,f)$ are respectively
\[
G_{red,0}=\left(\begin{array}{cc}\mu_{A}(p+1) & 0 \\0 & \mu_{B}(p-1)\end{array}\right)
\]
and
\[
G_{red,1}=c_n^2\left(\begin{array}{cc}(p+1)\big(\mu_{A}+\frac{\mu_{B}}{n}\big) & \sqrt{\frac{p^2-1}{n}}(\mu_A-\mu_B) \\ \sqrt{\frac{p^2-1}{n}}(\mu_A-\mu_B) & (p-1)\big(\frac{\mu_{A}}{n}+\mu_{B}\big)\end{array}\right).
\]
We recall
\[KL(P_{1}^{obs},P_{0}^{obs})=\frac{1}{2}\left(\log\Bigg(\frac{\det(G_{red,0})}{\det(G_{red,1})}\Bigg)-2+\text{Trace}(G_{red,0}^{-1}G_{red,1})\right).\]
see for instance Equation (A23) of \cite{Rasmussen}.
We have:
\[\det(G_{red,0})=\mu_A\mu_B(p^2-1)\]
and
\begin{align*}
\det(G_{red,1})&=c_n^4(p^2-1)\Bigg(\big(\mu_{A}+\frac{\mu_{B}}{n}\big)\big(\frac{\mu_{A}}{n}+\mu_{B}\big)-\frac{(\mu_A-\mu_B)^2}{n}\Bigg) \\
&= \mu_A\mu_B(p^2-1)c_n^4\Big(1+\frac{1}{n}\Big)^2\\
&=\mu_A\mu_B(p^2-1).
\end{align*}
Finally,
\[
G_{red,0}^{-1}=\left(\begin{array}{cc}\mu_{A}^{-1}(p+1)^{-1} & 0 \\0 & \mu_{B}^{-1}(p-1)^{-1}\end{array}\right)
\]
and
\[
G_{red,0}^{-1}G_{red,1}=c_n^2\left(\begin{array}{cc}\big(1+\frac{\mu_{B}}{\mu_{A}n}\big) & \sqrt{\frac{p-1}{n(p+1)}}\big(1-\frac{\mu_{B}}{\mu_{A}}\big) \\ -\sqrt{\frac{p+1}{n(p-1)}}\big(1-\frac{\mu_{A}}{\mu_{B}}\big)& \big(1+\frac{\mu_{A}}{\mu_{B}n}\big)\end{array}\right),
\]
which yields, with $c_n^2=n/(n+1)$,
\[KL(P_{1}^{obs},P_{0}^{obs})=\frac{1}{2}\left(-2+2c_n^2+\frac{c_n^2}{n}\Big(\frac{\mu_{A}}{\mu_{B}}+\frac{\mu_{B}}{\mu_{A}}\Big)\right)=\frac{1}{2(n+1)}\Big(\frac{\mu_{A}}{\mu_{B}}+\frac{\mu_{B}}{\mu_{A}}-2\Big).\]
Lemma~\ref{lemma:kullback} is proved.
%%%%%%%
\subsubsection{Proof of Lemma~\ref{lemma:trigo}}\label{sec:trigo}
Let $x\in (0,2\pi)$. We have:
\begin{align*}
\sum_{k=0}^{p-1}e^{ixk}&=\frac{1-e^{ixp}}{1-e^{ix}}
=\frac{e^{ixp/2}\big(e^{-ixp/2}-e^{ixp/2}\big)}{e^{ix/2}\big(e^{-ix/2}-e^{ix/2}\big)}
=e^{ix(p-1)/2}\frac{\sin(xp/2)}{\sin(x/2)}.
\end{align*}
Let $p\geq 3$. We take $x=2\pi/(p-1)$ which lies in $(0,2\pi)$. Considering the real and imaginary parts, we obtain:
$$\sum_{k=0}^{p-1}\cos(2\pi t_k)=\cos(\pi)\frac{\sin(\pi p/(p-1))}{\sin(\pi/(p-1))}=1$$ 
and
$$\sum_{k=0}^{p-1}\sin(2\pi t_k)=\sin(\pi)\frac{\sin(\pi p/(p-1))}{\sin(\pi/(p-1))}=0.$$ 
Similarly,  if $p\geq 4$, with $x=4\pi/(p-1)$ which lies in $(0,2\pi)$, we have
$$\sum_{k=0}^{p-1}\sin(2\pi t_k)=0.$$
This result remains true for $p=3$. Now, we have, for $p\geq 4$,
$$\sum_{k=0}^{p-1}\cos^2(2\pi t_k)=\sum_{k=0}^{p-1}\frac{\cos(4\pi t_k)+1}{2}=\frac{\frac{\sin(2\pi p/(p-1))}{\sin(2\pi/(p-1))}+p}{2}=\frac{p+1}{2}$$
and 
$$\sum_{k=0}^{p-1}\sin^2(2\pi t_k)=\sum_{k=0}^{p-1}\Big(1-\cos^2(2\pi t_k)\Big)=\frac{p-1}{2}.$$
Lemma~\ref{lemma:trigo} is proved.
%%%%%%%%%%%%%%%%%%%%%
\subsection{Proof of Proposition~\ref{prop:borne_inf_fpp}}\label{proof:borneinfp}
The proof is based on Assouad's Lemma and follows the general scheme described in \citet[Sections 2.6 and 2.7]{tsybakov_introduction_2009}. 
Let 
\[
\phi(t)=e^{-\frac1{1-t^2}}1_{(-1,1)}(t).
\]
We then define 
\[
\varphi(t)=\left\{
\begin{array}{cl}
	\phi(4t-1) &\text{ if }t\in[0,1/2),\\
	-\phi(4t+1) &\text{ if }t\in(-1/2,0],\\
	0&\text{ if }t\notin(-1/2,1).
\end{array}
\right.
\]
Both functions $\phi$ and $\varphi$ are $C^\infty$ on $\mathbb R$ with bounded support, then are $\alpha$-H\"older continuous, for all $\alpha>0$. The function $\varphi$ has its support included in $(-1/2,1/2)$ and verifies $\int_{-1/2}^{1/2}\varphi(t)dt=0$. We note $ L_\alpha$ such that, for all $t,u\in\mathbb R$, 
\[
|\varphi(t)-\varphi(u)|\leq L_\alpha |t-u|^{\alpha}.
\]
Let us now define test eigenfunctions. For $\bomega=(w_0,\ldots,w_{p-1})\in\{0,1\}^p$, we set
$$\eta_{1,\bomega}^*(t)=C_{\bomega}\left(\gamma+\sum_{k=0}^{p-1}\omega_k \left(p^{-\alpha}\varphi\left(p(t-t_k)-1/2\right)\right)\right),$$
with $C_{\bomega}$ and $\gamma>0$ two positive constants to be specified later. To be an eigenfunction, $\eta_{1,\bomega}^*$ has to be of norm 1, which writes
\begin{eqnarray*}
	\|\eta_{1,\bomega}^*\|^2&=&C_{\bomega}^2\int_0^1\left(\gamma+\sum_{k=0}^{p-1}\omega_k\left(p^{-\alpha}\varphi(p(t-t_k)-1/2)\right)\right)^2dt\\
	&=&C_{\bomega}^2\left(\gamma^2+2\gamma\sum_{k=0}^{p-1}\omega_k\left(p^{-\alpha}\int_0^1\varphi(p(t-t_k)-1/2)dt\right)\right.\\
	&&\left.+\int_0^1\left(\sum_{k=0}^{p-1}\omega_k\left(p^{-\alpha}\varphi(p(t-t_k)-1/2)\right)\right)^2dt\right).
\end{eqnarray*}
Using successively that the support of $\varphi$ is in $(-1/2,1/2)$ and that $\int_{-1/2}^{1/2}\varphi(t)dt=0$, we have
\[
\int_0^1\varphi(p(t-t_k)-1/2)dt=\int_{t_k}^{t_{k+1}}\varphi(p(t-t_k)-1/2)dt=p^{-1}\int_{-1/2}^{1/2}\varphi(t)dt=0,
\]
and 
\[
\int_0^1\left(\sum_{k=0}^{p-1}\omega_k\varphi(p(t-t_k)-1/2)\right)^2dt=\sum_{k=0}^{p-1}\omega_k\int_0^1\varphi^2(p(t-t_k)-1/2)dt=p^{-1}\sum_{k=0}^{p-1}\omega_k\|\varphi\|^2.
\]
This implies that 
\begin{eqnarray*}
	\|\eta_{1,\bomega}^*\|^2
	&=&C_{\bomega}^2\left(\gamma^2+p^{-2\alpha-1}\|\varphi\|^2\sum_{k=0}^{p-1}\omega_k\right).
\end{eqnarray*}
We then fix the quantity 
\[
C_{\bomega}=\left(\gamma^2+p^{-2\alpha-1}\|\varphi\|^2\sum_{k=0}^{p-1}\omega_k\right)^{-1/2},
\]
so that $\|\eta_{1,\bomega}^*\|=1$ and observe that $C_{\bomega}$ verifies
\begin{equation*}%\label{eq:Comega}
	\left(\gamma^2+\|\varphi\|^2\right)^{-1/2}\leq\left(\gamma^2+p^{-2\alpha}\|\varphi\|^2\right)^{-1/2}\leq C_{\bomega}\leq\gamma^{-1}.
\end{equation*}
We now define the associated distribution of our observations: for $\xi$ a centered random variable with variance 1 and $\mu_{1,\bomega}^*=\frac{L}{2L_{\alpha}^2C_\bomega^2}$, %i.e. $C_\bomega^2=\frac{L}{2L_{\alpha}^2\mu_{1,\bomega}^*}$ 
we set
\begin{equation}\label{Z:def}
Z_{\bomega}(t)=\sqrt{\mu_{1,\bomega}^*}\xi\eta_{1,\bomega}^*(t).
\end{equation}
%\begin{remark}
%Observe that subsequent computations remain valid if we just assume that $\xi$ is centered and of variance 1.
%\end{remark}
Let $P_{\bomega}^Z$ be the distribution of $Z_\bomega$. We have that $P_{\bomega}^Z\in\mathcal R_\alpha(L)$ since
\begin{eqnarray*}
	\int_{C([0,1])}(z(t)-z(s))^2dP_{\bomega}^Z(z)&=&\mathbb E[(Z_{\bomega}(t)-Z_{\bomega}(s))^2]=\mu_{1,\bomega}^*(\eta_{1,\bomega}(t)-\eta_{1,\bomega}(s))^2\mathbb E[\xi^2]\\
	&=&\mu_{1,\bomega}^*(\eta_{1,\bomega}(t)-\eta_{1,\bomega}(s))^2\\
	&=&\mu_{1,\bomega}^*C_{\bomega}^2\left(\sum_{k=0}^{p-1}\omega_kp^{-\alpha}(\varphi(p(t-t_k)-1/2)-\varphi(p(s-t_k)-1/2))\right)^2. 
\end{eqnarray*}
Then, using the properties of $\varphi$, we have two cases: 
\begin{itemize}
	\item If $s,t\in[t_\ell,t_{\ell+1}[$ for some $\ell\in\{0,\hdots, p-1\}$, 
	\begin{eqnarray*}
		\left(\sum_{k=0}^p\omega_kp^{-\alpha}(\varphi(p(t-t_\ell)-1/2)-\varphi(p(s-t_\ell)-1/2))\right)^2&&\\
		&&\hspace{-8cm}=\omega_\ell^2p^{-2\alpha}(\varphi(p(t-t_\ell)-1/2)-\varphi(p(s-t_\ell)-1/2))^2\\
		&&\hspace{-8cm}\leq p^{-2\alpha}L_{\alpha}^2|p(t-t_\ell)-p(s-t_\ell)|^{2\alpha}=L_{\alpha}^2|t-s|^{2\alpha}.
	\end{eqnarray*}
	\item If $s\in[t_\ell,t_{\ell+1}[$ and  $t\in[t_{\ell'},t_{\ell'+1}[$ with $\ell\neq\ell'$,
	\begin{eqnarray*}
		\left(\sum_{k=0}^p\omega_kp^{-\alpha}(\varphi(p(t-t_k)-1/2)-\varphi(p(s-t_k)-1/2))\right)^2&&\\
		&&\hspace{-8cm}=\omega_\ell^2p^{-2\alpha}|\varphi(p(t-t_\ell)-1/2)-\varphi(p(s-t_\ell)-1/2)|^2.\\
		&&\hspace{-6cm}+\omega_{\ell'}^2p^{-2\alpha}|\varphi(p(t-t_{\ell'})-1/2)-\varphi(p(s-t_{\ell'})-1/2)|^2\\
		&&\hspace{-8cm}\leq 2L_{\alpha}^2|t-s|^{2\alpha}.
	\end{eqnarray*}
\end{itemize}
Finally
\[
\int_{C([0,1])}(z(t)-z(s))^2dP_{\bomega}(z)\leq 2\mu_{1,\bomega}^*C_{\bomega}^2L_\alpha^2|t-s|^{2\alpha}=L|t-s|^{2\alpha}. 
\]
This allows to deduce that 
\[
\inf_{\widehat\eta_1}\sup_{P_Z\in\mathcal R_\alpha(L)}\mathbb E[\|\widehat\eta_1-\eta_1^*\|^2]\geq\inf_{\widehat\eta_1}\sup_{\bomega\in\{0,1\}^{p}}\mathbb E[\|\widehat\eta_1-\eta_{1,\bomega}^*\|^2] ,
\] 
and the aim of what follows is to prove a lower bound for $\mathbb E[\|\widehat\eta_1-\eta_{1,\bomega}^*\|^2]$. 

Let $\widehat\eta_1$ an estimator and 
\[
\widehat\bomega\in {\arg\min}_{\bomega\in\{0,1\}^p} \|\widehat\eta_1-\eta_{1,\bomega}^*\|^2,
\]
we have 
\[
\|\widehat\eta_1-\eta_{1,\widehat\bomega}^*\|\geq\frac12\|\eta_{1,\widehat\bomega}^*-\eta_{1,\bomega}^*\|. 
\]
Now, still from the support properties of $\varphi$,
\begin{eqnarray*}
	\|\eta_{1,\widehat\bomega}^*-\eta_{1,\bomega}^*\|^2&&\\
	&&\hspace{-3cm}=\sum_{k=0}^{p-1}\int_{t_{k}}^{t_{k+1}}\left(C_{\widehat\bomega}(\gamma+\widehat\omega_{k}p^{-\alpha}\varphi(p(t-t_{k})-1/2))-C_{\bomega}(\gamma+\omega_{k}p^{-\alpha}\varphi(p(t-t_{k})-1/2))\right)^2dt\\
	&&\hspace{-3cm}=p^{-1}\sum_{k=0}^{p-1}\int_{-1/2}^{1/2}\left(C_{\widehat\bomega}(\gamma+\widehat\omega_{k}p^{-\alpha}\varphi(u))-C_{\bomega}(\gamma+\omega_{k}p^{-\alpha}\varphi(u))\right)^2du\\
	&&\hspace{-3cm}=(C_{\widehat\bomega}-C_{\bomega})^2\gamma^2+\|\varphi\|^2p^{-2\alpha-1}\sum_{k=0}^{p-1}(C_{\widehat\bomega}\widehat\omega_{k}-C_{\bomega}\omega_{k})^2\geq\|\varphi\|^2p^{-2\alpha-1}\sum_{k=0}^{p-1}(C_{\widehat\bomega}\widehat\omega_{k}-C_{\bomega}\omega_{k})^2\\
	%&&\hspace{-3cm}\geq\|\varphi\|^2p^{-2\alpha-1}\left(\min\{C_{\widehat\bomega}^2,C_\bomega^2\}\sum_{k=0}^{p-1}\mathbf 1_{\{\widehat\omega_k\neq\omega_k\}}^2+(C_{\widehat\bomega}-C_{\bomega})^2\sum_{k=0}^{p-1}\mathbf 1_{\{\widehat\omega_k=\omega_k=1\}}^2\right)\\
	&&\hspace{-3cm}\geq\|\varphi\|^2p^{-2\alpha-1}\min\{C_{\widehat\bomega}^2,C_\bomega^2\}\sum_{k=0}^{p-1}\mathbf 1_{\{\widehat\omega_k\neq\omega_k\}}\\
	&&\hspace{-3cm}%\geq\|\varphi\|^2p^{-2\alpha-1}\min\{C_{\widehat\bomega}^2,C_\bomega^2\}\rho(\bomega,\widehat\bomega)
	\geq (\gamma^2+\|\varphi\|^2)^{-1}\|\varphi\|^2p^{-2\alpha-1}\rho(\widehat\bomega,\bomega),
\end{eqnarray*}
where $\rho(\bomega,\bomega')=\sum_{k=0}^{p-1}\mathbf 1_{\omega_k\neq\omega_k'}$ is the Hamming distance on $\{0,1\}^p$.

Combining all the inequalities above, we have the existence of a constant $\tilde c=\|\varphi\|^2/(4(\gamma^2+\|\varphi\|^2))$ such that
\begin{eqnarray*}
	\inf_{\widehat\eta_1}\sup_{P_Z\in\mathcal R_\alpha(L)}\mathbb E[\|\widehat\eta_1-\eta_{1,\bomega}^*\|^2]&\geq& \tilde cp^{-2\alpha-1}\inf_{\widehat\bomega}\max_{\bomega\in\{0,1\}^{p}}\mathbb E[\rho(\widehat\bomega,\bomega)].
	%&\geq & c(\|\varphi\|^2p^{-2\alpha}+n^{-1})\inf_{\widehat\bomega}\max_{\bomega\in\{0,1\}^{m}}\mathbb E[\rho(\widehat\bomega,\bomega)]
\end{eqnarray*}
By Assouad's lemma (see e.g. Tsybakov, 2009, Theorem 2.12),  there exists a constant $c>0$ such that 
\begin{equation}\label{inf-KL}
	\inf_{\widehat\bomega}\max_{\bomega\in\{0,1\}^{p}}\mathbb E[\rho(\widehat\bomega,\bomega)]\geq cp,
\end{equation}
provided we are able to prove that for some constant $K_{\max}\geq 0$,
\[
KL((P_{\bomega}^{obs})^{\otimes n},(P_{0}^{obs})^{\otimes n})\leq K_{\max}, \text{ for all } \bomega\in\{0,1\}^p,
\]
where $P_{\bomega}^{obs}$ is the law of the random vector 
\[\mathbf Y^{obs}_{\bomega}:=(Y_{\bomega}(t_0),\hdots,Y_{\bomega}(t_{p-1}))
\] such that 
\[
Y_{\bomega}(t_j)=Z_{\bomega}(t_j)+\varepsilon_j\] 
with $\varepsilon_0,\hdots,\varepsilon_{p-1}\sim_{i.i.d.}\mathcal N(0,\sigma^2)$ and $KL(P,Q)$ is the Kullback-Leibler divergence between two measures $P$ and $Q$. 
In \eqref{inf-KL}, the constant $c$ only depends on $K_{\max}$.
We observe that, for all $\bomega\in\{0,1\}^p$, for all $j=0,\hdots,p-1$,
\[
Y_\bomega(t_j)=Z_\bomega(t_j)+\varepsilon_j=\sqrt{\mu_{1,\bomega}^*}\xi\eta_{1,\bomega}^*(t_j)+\varepsilon_j. 
\]
Now
\[
\eta_{1,\bomega}^*(t_j)=C_\bomega\left(\gamma+\sum_{k=0}^{p-1}\omega_k(p^{-\alpha}\varphi(p(t_j-t_k)-1/2))\right)=C_\bomega\gamma,
\]
since $\varphi((p(t_j-t_k)-1/2)=\varphi(-1/2)=0$ if $j=k$ and $\varphi((p(t_j-t_k)-1/2)=0$ if $j\neq k$ by the support properties of $\varphi$ and the fact that 
\[p(t_j-t_k)-1/2=\frac{p}{p-1}(j-k)-1/2\geq \frac{p}{p-1}-1/2\geq 1/2
\] if $j>k$ and $p(t_j-t_k)-1/2\leq 1/2$ if $j<k$. 
Hence 
\[
Y_{\bomega}(t_j)=\sqrt{\mu_{1,\bomega}^*}\xi C_\bomega\gamma+\varepsilon_j=\frac{\gamma \sqrt{L}}{L_{\alpha}\sqrt{2}}\xi+\varepsilon_j
\]
and the distribution of  $\mathbf Y^{obs}_{\bomega}$ does not depend on ${\bomega}$. Therefore,
\[
KL((P_{\bomega}^{obs})^{\otimes n},(P_{0}^{obs})^{\otimes n})=nKL(P_{\bomega}^{obs},P_{0}^{obs})=0.
\]
%%%%%%%%%%%%%%%%%%%%%
%%%%%%%%%%%%%%%%%%%%%
\section{Proof of Theorems~\ref{theo:GeneralEsp} and \ref{theo:Generalproba}}
\label{proof:theo2-3}
%%%%%%%%%%%%%%%%%%%%%
\subsection{Preliminary result}
The proof of Theorems~\ref{theo:GeneralEsp} and \ref{theo:Generalproba} is based on Bosq inequalities stated in the following theorem.
\begin{theorem}[\citealt{Bosq}]
	\label{preBosq}
	Let $\Gamma$ and $\widehat\Gamma$ be two linear compact operators on $\mathbb L^2([0,1])$.	We denote by 
	\begin{equation*}%\label{eq:spectral}
		\Gamma=\sum_{d=1}^\infty\mu_d^*\eta_d^*\otimes\eta_d^*\quad\text{ and }\quad\widehat\Gamma=\sum_{d=1}^\infty\widehat\mu_d\widehat\eta_d\otimes\widehat\eta_d
	\end{equation*}
	their spectral decomposition with the eigenvalues $(\mu_d^*)_{d\geq 1}$ and $(\widehat\mu_d)_{d\geq 1}$ sorted in decreasing order. Then 
	\begin{equation}\label{preBosq inequality 1}
		|\widehat\mu_d-\mu_d^*|\leq \vvvert\widehat\Gamma-\Gamma\vvvert.
	\end{equation}
%	where $\vvvert\cdot\vvvert$ is the operator norm associated to $\|\cdot\|$ defined by $\vvvert T\vvvert=\sup_{f\in\mathbb L^2([0,1]), \|f\|=1}\|Tf\|$ for all continuous operator $T\in\mathcal L(\mathbb L^2([0,1]))$.
	Suppose moreover that, for $d\geq 1$, the eigenspace associated to the eigenfunction $\eta_d^*$ is one-dimensional and denote, to avoid sign confusion, $\eta_{\pm,d}^*=\text{sign}(\langle\widehat{\eta}_{\phi,d},\eta_d^*\rangle)\times\eta_d^*.$ Then, we have
	\begin{equation}\label{preBosq inequality 2}
		\|\widehat\eta_d-\eta_{\pm,d}^*\|\leq b_d^{1/2}\vvvert\widehat\Gamma-\Gamma\vvvert,
	\end{equation}
	where 
	\[
	b_1=8(\mu_1^*-\mu_2^*)^{-2}
	\]
	and for any $d\in\{2,\ldots,\mathcal D\}$
	\[
	b_d=8/\min(\mu_d^*-\mu_{d+1}^*,\mu_{d-1}^*-\mu_d^*)^2.
	\]
\end{theorem} 
The proof of Theorem~\ref{preBosq} comes directly from \citet[Lemma 4.2, p.~103]{Bosq} for the upper bound~\eqref{preBosq inequality 1} on the eigenvalues and \citet[Lemma 4.3, p.104]{Bosq} for the upper bound~\eqref{preBosq inequality 2} on the eigenfunctions. 

We remark that $\Gamma$, $\widehat\Gamma$ and $\widehat\Gamma_\phi$ are integral operators with kernel respectively $K$
$$
\widehat K(s,t)=\frac1n\sum_{i=1}^n Z_i(t)Z_i(s), \quad  (s,t) \in [0,1]^2.
$$
and 
\begin{equation*}%\label{eq:Kestim}
\widehat{K}_\phi(s,t)=\frac1n\sum_{i=1}^n \widetilde Y_i(t)\widetilde Y_i(s),\quad  (s,t) \in [0,1]^2.
\end{equation*}

We use the previous result to establish the following proposition.
\begin{proposition}
	\label{prop-inter}
	Setting $K_\phi=E[\widehat{K}_\phi]$, we have
	\begin{equation}\label{ubinter}
		\|\widehat{\eta}_{\phi,d}-\eta_{\pm,d}^*\|^2\leq 5b_d \left[\vvvert\widehat{\Gamma}_\phi-\Gamma_\phi\vvvert^2+\vvvert\Pi_D\Gamma\Pi_D-\Gamma\vvvert^2+\frac{\sigma^4}{p^2}+A_p^{(K)}(\phi,D)+A_p^{(\sigma)}(\phi,D)\right].
	\end{equation}
\end{proposition}
\begin{proof}[of Proposition~\ref{prop-inter}]
	In the sequel, we denote $\Gamma_\phi=E[\widehat{\Gamma}_\phi]$ and remark that $\Gamma_\phi$ is an integral operator with kernel $K_\phi$. 
	\begin{eqnarray}
		K_\phi(s,t)&=&\sum_{\lambda,\lambda'\in\Lambda_D}\frac1{p^2}\sum_{h,h'=0}^{p-1}K(t_h,t_{h'})\phi_\lambda(t_h)\phi_{\lambda'}(t_{h'})\phi_\lambda(s)\phi_{\lambda'}(t)\nonumber\\
		&&+\frac{\sigma^2}{p^2}\sum_{\lambda,\lambda'\in\Lambda_D}\sum_{h=0}^{p-1}\phi_{\lambda}(t_h)\phi_{\lambda'}(t_{h})\phi_{\lambda}(s)\phi_{\lambda'}(t)\nonumber\\
		&=&\Pi_{S_D^2}K(s,t)+\frac{\sigma^2}{p}\sum_{\lambda\in\Lambda_D}\phi_{\lambda}(s)\phi_{\lambda}(t)+R^{(K)}(s,t)+R^{(\sigma)}(s,t),\label{eq:decompKtilde}
	\end{eqnarray}
	where $\Pi_{S_D^2}$ is the orthogonal projection onto $S_D^2=\text{span}\{(s,t)\mapsto \phi_\lambda(s)\phi_{\lambda'}(t), \lambda,\lambda'\in\Lambda_D\}$,
	\begin{eqnarray*}
		R^{(K)}(s,t)&=&\sum_{\lambda,\lambda'\in\Lambda_D}\frac1{p^2}\sum_{h,h'=0}^{p-1}K(t_h,t_{h'})\phi_\lambda(t_h)\phi_{\lambda'}(t_{h'})\phi_\lambda(s)\phi_{\lambda'}(t)-\Pi_{S_D^2}K(s,t)\\
		&=&\sum_{\lambda,\lambda'\in\Lambda_D}\left(\frac1{p^2}\sum_{h,h'=0}^{p-1}K(t_h,t_{h'})\phi_\lambda(t_h)\phi_{\lambda'}(t_{h'})-\int_0^1\int_0^1K(s,t)\phi_{\lambda}(s)\phi_{\lambda'}(t)dsdt\right)\phi_\lambda(s)\phi_{\lambda'}(t)
	\end{eqnarray*}
	and
	\begin{eqnarray*}
		R^{(\sigma)}(s,t)
		&=&\frac{\sigma^2}{p^2}\sum_{\lambda,\lambda'\in\Lambda_D}\sum_{h=0}^{p-1}\phi_{\lambda}(t_h)\phi_{\lambda'}(t_{h})\phi_{\lambda}(s)\phi_{\lambda'}(t)-\frac{\sigma^2}{p}\sum_{\lambda\in\Lambda_D}\phi_{\lambda}(s)\phi_{\lambda}(t)\\
		&=&\frac{\sigma^2}{p}\sum_{\lambda,\lambda'\in\Lambda_D}\left(\frac1p\sum_{h=0}^{p-1}\phi_{\lambda}(t_h)\phi_{\lambda'}(t_{h})-\1_{\{\lambda=\lambda'\}}\right)\phi_{\lambda}(s)\phi_{\lambda'}(t).
	\end{eqnarray*}
	Then, from the decomposition of the kernel $K_\phi$ given in Equation~\eqref{eq:decompKtilde}, we have for any fonction $f$ and any $t\in[0,1]$, 
	\begin{align*}
		\Gamma_\phi(f)(t)&=\int_0^1 K_\phi(s,t)f(s)ds\\
		&=\int_0^1\Pi_{S_D^2}K(s,t)f(s)ds+\frac{\sigma^2}{p}\sum_{\lambda\in\Lambda_D}\int_0^1\phi_{\lambda}(s)f(s)ds\,\phi_{\lambda}(t)+T^{(K)}(f)(t)+T^{(\sigma)}(f)(t)\\
		&=\int_0^1\Pi_{S_D^2}K(s,t)f(s)ds+\frac{\sigma^2}{p}\Pi_D(f)(t)+T^{(K)}(f)(t)+T^{(\sigma)}(f)(t),
	\end{align*}
	where $\Pi_D$ is the orthogonal projection onto $S_D={\rm span}\{\phi_\lambda,\lambda\in\Lambda_D\}$ and $T^{(K)}$ (resp. $T^{(\sigma)}$) is the integral operator associated to the kernel $R^{(K)}$ (resp. $R^{(\sigma)}$):
	\[T^{(K)}(t):=\int_0^1 R^{(K)}(s,t)f(s)ds,\quad T^{(\sigma)}(f)(t):=\int_0^1 R^{(\sigma)}(s,t)f(s)ds.\]
	Now, 
	\begin{align}\label{pi}
		\int_0^1\Pi_{S_D^2}K(s,t)f(s)ds&=\sum_{\lambda,\lambda'\in\Lambda_D}\int_0^1\int_0^1\int_0^1K(u,v)\phi_{\lambda}(u)\phi_{\lambda'}(v)dudv\,\phi_{\lambda}(s)\phi_{\lambda'}(t)f(s)ds\nonumber\\
		&=\sum_{\lambda,\lambda'\in\Lambda_D}\langle \phi_{\lambda},f\rangle\int_0^1\int_0^1 K(u,v)\phi_{\lambda}(u)\phi_{\lambda'}(v)dudv\, \phi_{\lambda'}(t)\nonumber\\
		&=\sum_{\lambda,\lambda'\in\Lambda_D}\langle \phi_{\lambda},f\rangle \langle \Gamma(\phi_\lambda),\phi_{\lambda'}\rangle \, \phi_{\lambda'}(t)\nonumber\\
		&=\sum_{\lambda'\in\Lambda_D}\langle\Gamma(\sum_{\lambda\in\Lambda_D}\langle \phi_{\lambda},f\rangle\phi_{\lambda}),\phi_{\lambda'}\rangle\, \phi_{\lambda'}(t)\nonumber\\
		&=\Pi_D(\Gamma(\Pi_D(f)))(t).
	\end{align}
	Hence, we obtain:
	\[
	\Gamma_\phi=\Pi_D\Gamma\Pi_D+\frac{\sigma^2}p\Pi_D+T^{(K)}+T^{(\sigma)}.
	\]
	Now, 
	%\begin{eqnarray}
	%|\widehat{\widetilde\mu}_\ell-\mu_\ell^*|&\leq&\|\widehat{\widetilde\Gamma}-\Gamma\|_\infty\nonumber\\
	%%&\leq &\|\widehat{\widetilde\Gamma}-\widetilde\Gamma\|_\infty+\|\tilde\Gamma'-\Gamma\|_\infty+\frac{\sigma^2}p\mathbf 1_{\ell\leq D}\\
	%&\leq &\|\widehat{\widetilde\Gamma}-\widetilde\Gamma\|_\infty+\|\Pi_D\Gamma\Pi_D-\Gamma\|_\infty+\frac{\sigma^2}p+\|T^{(K)}\|_\infty+\|T^{(\sigma)}\|_\infty\label{eq:upper1vp}
	%\end{eqnarray}
	%and, similarly, for the eigenfunctions, 
	since the eigenvalues $(\mu_d^*)_{d\geq 1}$ are all distincts, the eigenspace associated to the eigenvalue $\mu_d^*$ is one-dimensional and we can apply Theorem~\ref{preBosq} to the operators $\Gamma$ and $\widehat{\Gamma}_\phi$, which yields 
	\begin{eqnarray}\label{eq:upper1fp}
		\hspace{-0.5cm}\|\widehat\eta_{\phi,d}-\eta_{\pm,d}^*\|&\leq& b_d^{1/2}\vvvert\widehat{\Gamma}_\phi-\Gamma\vvvert
		\nonumber\\
		&\leq &b_d^{1/2}\left(\vvvert\widehat{\Gamma}_\phi-\Gamma_\phi\vvvert+\vvvert\Pi_D\Gamma\Pi_D-\Gamma\vvvert\right.\\
		&&\hspace{3cm}	\left.	+\frac{\sigma^2}p+\vvvert T^{(K)}\vvvert+\vvvert T^{(\sigma)}\vvvert\right).\nonumber
	\end{eqnarray}
	In the previous inequality, we have used that $\vvvert\Pi_D\vvvert=1$. We now control each term of the previous inequality. For this purpose, introducing $\|\cdot\|_{HS}$, the Hilbert-Schmidt norm of an operator defined by $\|T\|_{HS}^2=\sum_{\lambda\in\Lambda}\|Te_\lambda\|^2$ where $(e_\lambda)_{\lambda\in\Lambda}$ is an orthonormal basis of $\mathbb L^2$ (recall that the Hilbert-Schmidt norm is independent of the choice of the basis), we have, for all operator $T:\mathbb L^2\mapsto\L^2$, 
	$\vvvert T\vvvert\leq \|T\|_{HS}$ since
	\[
	\vvvert T\vvvert^2=\sup_{f\in\mathbb L^2,f\neq 0}\frac{\|Tf\|^2}{\|f\|^2}
	\]
	and, by Cauchy-Schwarz's Inequality,
	\begin{eqnarray*}
		\|Tf\|^2&=&\sum_{\lambda\in\Lambda}\langle Tf,e_\lambda\rangle^2=\sum_{\lambda\in\Lambda}\left(\sum_{\lambda'\in\Lambda}\langle f,e_{\lambda'}\rangle\langle Te_{\lambda'},e_\lambda\rangle\right)^2\\
		&\leq& \sum_{\lambda\in\Lambda}\left(\sum_{\lambda'\in\Lambda}\langle f,e_{\lambda'}\rangle^2\sum_{\lambda'\in\Lambda}\langle Te_{\lambda'},e_\lambda\rangle^2\right)=\|f\|^2\sum_{\lambda'\in\Lambda}\|Te_{\lambda'}\|^2=\|f\|^2\|T\|_{HS}^2.
	\end{eqnarray*}
	Moreover, we also remark that if $T$ is a kernel operator associated to a kernel $R$, 
	\begin{align*}
		\|T\|^2_{HS}&=\sum_{\lambda\in\Lambda}\|Te_\lambda\|^2=\sum_{\lambda\in\Lambda}\left\|\int_0^1R(s,\cdot)e_\lambda(s)ds\right\|^2\\
		&=\sum_{\lambda\in\Lambda}\int_0^1\left(\int_0^1R(s,t)e_\lambda(s)ds\right)^2dt=\int_0^1\int_0^1R^2(s,t)dsdt=\|R\|^2.
	\end{align*}
	In addition, if the kernel $R\in S_D^2$, i.e. if there exists a matrix $G=(G_{\lambda,\lambda'})_{\lambda,\lambda'\in\Lambda_D}$ such that 
	\[
	R(s,t)=\sum_{\lambda,\lambda'\in \Lambda_D}G_{\lambda,\lambda'}\phi_\lambda(s)\phi_{\lambda'}(t), 
	\] we have $\|R\|_{L^2}=\left\|G\right\|_F,$
	where, for a matrix $G$, 
	\[\|G\|_F=\sqrt{Tr(G^TG)}=\Big(\sum_{\lambda,\lambda'\in \Lambda_D}G_{\lambda,\lambda'}^2\Big)^{1/2}\] 
	is the Frobenius norm of the matrix $G$. The fourth and fifth terms of Equation~\eqref{eq:upper1fp} are then bounded by the squared Frobenius norm of the associated matrices and we obtain 
	\[
	\|\widehat{\eta}_{\phi,d}-\eta_{\pm,d}^*\|^2\leq 5b_d \left[\vvvert\widehat{\Gamma}_\phi-\Gamma_\phi\vvvert^2+\vvvert\Pi_D\Gamma\Pi_D-\Gamma\vvvert^2+\frac{\sigma^4}{p^2}+A_p^{(K)}(\phi,D)+A_p^{(\sigma)}(\phi,D)\right].
	\]
	Proposition~\ref{prop-inter} is proved.
\end{proof}

To end the proof of Theorems~\ref{theo:GeneralEsp} and \ref{theo:Generalproba}, it remains to deal with  the stochastic term $\vvvert\widehat{\Gamma}_\phi-\Gamma_\phi\vvvert^2$, still bounded by using the Frobenius norm:
\[
\vvvert\widehat{\Gamma}_\phi-\Gamma_\phi\vvvert^2
\leq \|{\widehat G}_\phi- G_\phi\|_F^2,
\]
where
\[
{\widehat G}_\phi:=\left(\frac1n\sum_{i=1}^n\widetilde Y_{i,\lambda}\widetilde Y_{i,\lambda'}\right)_{\lambda,\lambda'\in\Lambda_D}
\]
and $G_\phi=E[\widehat{G}_\phi]$. The upper bound of $E[\|{\widehat G}_\phi- G_\phi\|_F^2]$ gives Theorem~\ref{theo:GeneralEsp}, whereas Theorem~\ref{theo:Generalproba} is deduced from the control in probability of  $\|{\widehat G}_\phi- G_\phi\|_F$ provided by Proposition~\ref{concentration} below.
%%%%%%%%%%%%%%%%%%%%%
\subsection{End of the proof of Theorem~\ref{theo:GeneralEsp}}
\begin{lemma}
	\label{Moment}
	Under Assumption~\ref{Ass:4}, we have:
	\[
	E[\|{\widehat G}_\phi- G_\phi\|_F^2]\leq  \frac{\max(C_1+3;6)}{n}\left(\sum_{\lambda\in\Lambda_D}\Big[\sigma_\lambda^2+s_\lambda^2\Big]\right)^2.
	\]
\end{lemma}
\begin{proof}[of Lemma~\ref{Moment}]
	We have
	\begin{align*}
		E[\|{\widehat G}_\phi- G_\phi\|_F^2]&=\sum_{\lambda,\lambda'\in\Lambda_D}E\left[\left(\frac1n\sum_{i=1}^n\Big[\widetilde Y_{i,\lambda}\widetilde Y_{i,\lambda'}-E[\widetilde Y_{i,\lambda}\widetilde Y_{i,\lambda'}]\Big]\right)^2\right]\\
		&=\sum_{\lambda,\lambda'\in\Lambda_D}\text{Var}\left(\frac1n\sum_{i=1}^n\widetilde Y_{i,\lambda}\widetilde Y_{i,\lambda'}\right)\\
		&\leq\frac{1}{n}\sum_{\lambda,\lambda'\in\Lambda_D}E[\widetilde Y_{1,\lambda}^2\widetilde Y_{1,\lambda'}^2]\\
		&\leq\frac{1}{n}\left(\sum_{\lambda\in\Lambda_D}(E[\widetilde Y_{1,\lambda}^4])^{1/2}\right)^2.
	\end{align*}
	Now, since $\tilde\varepsilon_{1\lambda}\sim{\mathcal N}(0,\sigma_\lambda^2)$, and $\tilde z_{1\lambda}=\frac1p\sum_{h=0}^{p-1}Z_1(t_h)\phi_\lambda(t_h)$, we have
	\begin{align*}
		E[\widetilde Y_{1,\lambda}^4]&=E[(\tilde z_{1,\lambda}+\tilde\varepsilon_{1,\lambda})^4]\\
		&=E[\tilde z_{1,\lambda}^4]+6E[\tilde z_{1,\lambda}^2]E[\tilde\varepsilon_{1,\lambda}^2]+ E[\tilde\varepsilon_{1,\lambda}^4]\\
		&\leq C_1(E[\tilde z_{1,\lambda}^2])^2+6E[\tilde z_{1,\lambda}^2]E[\tilde\varepsilon_{1,\lambda}^2]+ 3\sigma_\lambda^4\\
		&\leq (C_1+3)s_\lambda^4+ 6\sigma_\lambda^4
	\end{align*}
	and 
	\begin{align*}
		E[\|{\widehat G}_\phi- G_\phi\|_F^2]&\leq \frac{1}{n}\left(\sum_{\lambda\in\Lambda_D}\Big((C_1+3)s_\lambda^4+ 6\sigma_\lambda^4\Big)^{1/2}\right)^2\\
		&\leq  \frac{\max(C_1+3;6)}{n}\left(\sum_{\lambda\in\Lambda_D}\Big[\sigma_\lambda^2+s_\lambda^2\Big]\right)^2.
	\end{align*}
	This ends the proof of Lemma~\ref{Moment}.
\end{proof}
Combining the upper bound of the previous lemma with \eqref{ubinter} provides the stated result in Theorem~\ref{theo:GeneralEsp}  .
%%%%%%%%%%%%%%%%%%%%%
\subsection{End of the proof of Theorem~\ref{theo:Generalproba}}
To complete the  proof of Theorem~\ref{theo:Generalproba}, we need some technical lemmas. Before stating them, we recall that for all $i=1,\ldots,n$, we have set
\[
\widetilde Y_{i,\lambda}=\frac{1}{p}\sum_{h=0}^{p-1}Y_i(t_h)\phi_\lambda(t_h),\quad \tilde z_{i,\lambda}=\frac{1}{p}\sum_{h=0}^{p-1}Z_i(t_h)\phi_\lambda(t_h),\quad \tilde\varepsilon_{i,\lambda}=\frac{1}{p}\sum_{h=0}^{p-1}\varepsilon_{i,h}\phi_\lambda(t_h)
\]
and
\[
s_\lambda^2 = {\rm Var}(\tilde z_{i,\lambda}),\quad \sigma_\lambda^2 = {\rm Var}(\tilde\varepsilon_{i,\lambda}).
\]
In the sequel, we consider
$\widetilde Y_i=(\widetilde Y_{i,\lambda})_{\lambda\in\Lambda_D}$, $\tilde z_i=(\tilde z_{i,\lambda})_{\lambda\in\Lambda_D}$ and $\tilde\varepsilon_i=(\tilde\varepsilon_{i,\lambda})_{\lambda\in\Lambda_D}.$

\begin{lemma}
	\label{SG1}
	Under Assumption~\ref{Ass:q}  ,
	for any $u\in\R^{|\Lambda_D|}$,
	\begin{equation}\label{SG0}
		\|u^T\tilde z_1\|^2_{\psi_2}\leq  C_2E[(u^T\tilde z_1)^2]. 
	\end{equation}
	If we consider
	$\tilde\varepsilon_1$ instead of $\tilde z_1$, Inequality~\eqref{SG0} holds with an absolute constant instead of $C_2$. Furthermore,
	\begin{equation}\label{trace}
		Tr\Big(E\big[\tilde z_1\tilde z_1^T\big]\Big)= \sum_{\lambda\in\Lambda_D}s_\lambda^2, \quad Tr\Big(E\big[\tilde\varepsilon_1\tilde\varepsilon_1^T\big]\Big)= \sum_{\lambda\in\Lambda_D}\sigma_\lambda^2.
	\end{equation}
\end{lemma}
\begin{proof}[of Lemma~\ref{SG1}]
	Since ${\bf Z_1}:=\{Z_1(t_0),\ldots,Z_1(t_{p-1})\}^T$ is a zero-mean sub-Gaussian vector, the vector $\tilde z_1$ is also a zero-mean sub-Gaussian vector. We have, for any $u\in\R^{|\Lambda_D|}$,
	\begin{align*}
		\|u^T\tilde z_1\|^2_{\psi_2}&=\left\|\sum_{\lambda\in\Lambda_D}u_\lambda \tilde z_{1,\lambda}\right\|^2_{\psi_2}\\
		&=\left\|\sum_{\lambda\in\Lambda_D}u_\lambda \times\frac{1}{p}\sum_{h=0}^{p-1}Z_1(t_h)\phi_\lambda(t_h)\right\|^2_{\psi_2}\\
		&=\left\| v^T{\bf Z_1}\right\|^2_{\psi_2},
	\end{align*}
	with $v=(v_h)_{h=0,\ldots,p-1}$ and $v_h:=\frac{1}{p}\sum_{\lambda\in\Lambda_D}u_{\lambda}\phi_\lambda(t_h).$ Therefore,
	\begin{align*}
		\|u^T\tilde z_1\|^2_{\psi_2}&\leq C_2E[(v^T{\bf Z_1})^2]\\
		&\leq C_2E\Big[\sum_{h,h'=0}^{p-1}v_hZ_1(t_h)Z_1(t_{h'})v_{h'}\Big]\\
		&\leq C_2\sum_{\lambda,\lambda'\in\Lambda_D}u_\lambda u_{\lambda'}\frac{1}{p^2}E\Big[\sum_{h,h'=0}^{p-1}\phi_\lambda(t_h)\phi_{\lambda'}(t_{h'})Z_1(t_h)Z_1(t_{h'})\Big]\\
		&\leq  C_2E[(u^T\tilde z_1)^2].
	\end{align*}
	%Using Lemma A3 of Bunea and Xiao (2015), we conclude that
	% $$\|\|\tilde z\|_{\ell_2}\|^2_{\psi_2}\leq 2c_0Tr\Big(E\big[\tilde z\tilde z^T\big]\Big)= 2c_0\sum_{\lambda\in\Lambda_D}E\big[\tilde z_\lambda^2\big]=2c_0\sum_{\lambda\in\Lambda_D}s_\lambda^2.$$
	Now, if we consider $\tilde\varepsilon_1$ instead of  $\tilde z_1$, setting ${\bf\varepsilon_1}:=(\varepsilon_{1,0},\ldots,\varepsilon_{1,p-1})^T,$ and using Section~5.2.3 and Lemma~5.24 of Vershynin (2012), 
	$$\left\| v^T{\bf\varepsilon_1}\right\|^2_{\psi_2}\leq C\sigma^2\|v\|_{\ell_2}^2= CE[(u^T\tilde \varepsilon_1)^2] ,$$
	with $C$ an absolute constant, and
	$$\|u^T\tilde \varepsilon_1\|^2_{\psi_2}\leq  CE[(u^T\tilde \varepsilon_1)^2].$$ 
	The equalities \eqref{trace} are obvious.
\end{proof}
Results of the previous lemma are useful for the following result.
\begin{lemma}
	\label{SG2}
	We denote $X=(X_{\lambda\lambda'})_{\lambda,\lambda'\in\Lambda_D}$ the matrix whose entries are
	$$X_{\lambda\lambda'}=\widetilde Y_{1,\lambda}\widetilde Y_{1,\lambda'}-E[\widetilde Y_{1,\lambda}\widetilde Y_{1,\lambda'}].$$
	Setting
	$$M_D:=\sum_{\lambda\in\Lambda_D}s_\lambda^2+\sum_{\lambda\in\Lambda_D}\sigma^2_\lambda= \sum_{\lambda\in\Lambda_D}\left(\frac1{p^2}\sum_{h,h'=0}^{p-1}K(t_h,t_{h'})\phi_\lambda(t_h)\phi_\lambda(t_{h'})+\frac{\sigma^2}{p^2}\sum_{h=0}^{p-1}\phi_\lambda^2(t_h)\right),$$
	under Assumption~\ref{Ass:q}, there existe an absolute constant $\bar C$ such that  for any $t\geq \bar C(C_2+1) M_D$,
	$$E\left[\exp(t^{-1}\|X\|_F )\right]\leq\exp(1).$$
\end{lemma}
\begin{proof}[of Lemma~\ref{SG2}]
	We have
	\begin{align*}
		\|X\|_F&=\|\widetilde Y_1\widetilde Y_1^T-E[\widetilde Y_1\widetilde Y_1^T]\|_F\\
		&\leq\|(\tilde z_1+\tilde \varepsilon_1)(\tilde z_1+\tilde \varepsilon_1)^T\|_F+\|E\big[(\tilde z_1+\tilde \varepsilon_1)(\tilde z_1+\tilde \varepsilon_1)^T\big]\|_F\\
		&\leq\|\tilde z_1\tilde z_1^T\|_F+\|\tilde\varepsilon_1\tilde\varepsilon_1^T\|_F+2\|\tilde z_1\tilde\varepsilon_1^T\|_F+\|E\big[\tilde z_1\tilde z_1^T\big]\|_F+\|E\big[\tilde\varepsilon_1\tilde\varepsilon_1^T\big]\|_F\\
		&\leq\|\tilde z_1\|_{\ell_2}^2+\|\tilde\varepsilon_1\|_{\ell_2}^2+2\|\tilde z_1\|_{\ell_2}\|\tilde\varepsilon_1\|_{\ell_2}+\|E\big[\tilde z_1\tilde z_1^T\big]\|_F+\|E\big[\tilde\varepsilon_1\tilde\varepsilon_1^T\big]\|_F.
	\end{align*}
	We also have
	\begin{align*}
		\|E\big[\tilde z_1\tilde z_1^T\big]\|_F^2&=\sum_{\lambda,\lambda'\in\Lambda_D}\big(E[\tilde z_{1,\lambda}\tilde z_{1,\lambda'}]\big)^2\\
		&\leq\sum_{\lambda\in\Lambda_D}\sum_{\lambda'\in\Lambda_D}E[\tilde z_{1,\lambda}^2]E[\tilde z_{1,\lambda'}^2]\\
		&\leq\Big(\sum_{\lambda\in\Lambda_D}s_\lambda^2\Big)^2.
	\end{align*}
	Therefore
	$$\|E\big[\tilde z_1\tilde z_1^T\big]\|_F\leq\sum_{\lambda\in\Lambda_D}s_\lambda^2$$
	and similarly,
	$$\|E\big[\tilde\varepsilon_1\tilde\varepsilon_1^T\big]\|_F\leq\sum_{\lambda\in\Lambda_D}\sigma_\lambda^2.$$
	We finally obtain
	\begin{align*}
		\|X\|_F&\leq 2\|\tilde z_1\|_{\ell_2}^2+2\|\tilde\varepsilon_1\|_{\ell_2}^2+M_D
	\end{align*} 
	and we have
	\begin{align*}
		E\left[\exp(t^{-1}\|X\|_F )\right]&\leq E\left[\exp(2t^{-1}\|\tilde z_1\|_{\ell_2}^2)\right]\times E\left[\exp(2t^{-1}\|\tilde\varepsilon_1\|_{\ell_2}^2)\right]\times\exp(t^{-1}M_D).
	\end{align*} 
	Then, using Lemma~\ref{SG1} and Proposition A.1. of Bunea and Xiao (2015), we obtain for $C_*$ and $c_*$ two absolute positive constants, if $t>c_*(4C_2+1)\sum_{\lambda\in\Lambda_D}s_\lambda^2$,
	\begin{align*}
		E\left[\exp(2t^{-1}\|\tilde z_1\|_{\ell_2}^2)\right]&\leq E\left[\exp\left(2t^{-1}\Big(\|\tilde z_1\|_{\ell_2}^2-\sum_{\lambda\in\Lambda_D}s_\lambda^2\Big)\right)\right]\times\exp\left(2t^{-1}\sum_{\lambda\in\Lambda_D}s_\lambda^2\right)\\
		&\leq \exp\left(C_*\left(\frac{(4C_2+1)\sum_{\lambda\in\Lambda_D}s_\lambda^2}{t}\right)^2+2t^{-1}\sum_{\lambda\in\Lambda_D}s_\lambda^2 \right).
	\end{align*} 
	Similarly, for $t$ larger than $\sum_{\lambda\in\Lambda_D}\sigma_\lambda^2$ up to a multiplicative absolute constant, 
	$$E\left[\exp(2t^{-1}\|\tilde\varepsilon_1\|_{\ell_2}^2)\right]\leq \exp\left(C_{**}\left(\frac{\sum_{\lambda\in\Lambda_D}\sigma_\lambda^2}{t}\right)^2+2t^{-1}\sum_{\lambda\in\Lambda_D}\sigma_\lambda^2 \right),$$
	where $C_{**}$ is an absolute constant.  This ends the proof of the lemma.
\end{proof}
The following proposition controls the term  $\|{\widehat G}_\phi- G_\phi\|_F$ as required to complete the proof of Theorem~\ref{theo:Generalproba}  .
\begin{proposition}
	\label{concentration}
	We assume that Assumption~\ref{Ass:q}   is satisfied. For $\gamma>0$, with probability larger than $1-2\exp(-1/64\min(\gamma^2,16\gamma\sqrt{n})),$
	$$\|{\widehat G}_\phi- G_\phi\|_F\leq \frac{ \bar C(e^{1/2}+\gamma)(C_2+1)}{\sqrt{n}}\sum_{\lambda\in\Lambda_D}\Big[\sigma_\lambda^2+s_\lambda^2\Big],$$
	where $\bar C$ is an absolute constant.
\end{proposition}
\begin{proof}[of Proposition~\ref{concentration}]
	We apply Theorem~4.1 of Juditsky and Nemiroski (2008) with $\alpha=1$, since $(\R^{|\Lambda_D|^2},\|\cdot\|_F)$ is 1-smooth (see Definition 2.1 of Juditsky and Nemiroski (2008)). Since Lemma~\ref{SG2} gives for
	$t\geq \bar C(C_2+1)\sum_{\lambda\in\Lambda_D}\Big[\sigma_\lambda^2+s_\lambda^2\Big]$,
	$$E\left[\exp(t^{-1}\|X\|_F )\right]\leq\exp(1),$$
	for $\gamma>0$, with probability larger than $1-2\exp(-1/64\min(\gamma^2,16\gamma\sqrt{n})),$ 
	\begin{align*}
		\|{\widehat G}_\phi- G_\phi\|_F&\leq 
		\frac{\bar C(e^{1/2}+\gamma)(C_2+1)}{\sqrt{n}}\sum_{\lambda\in\Lambda_D}\Big[\sigma_\lambda^2+s_\lambda^2\Big].
	\end{align*}
	Proposition~\ref{concentration} is proved.
\end{proof}
Plugging the upper bound of Proposition~\ref{concentration} in \eqref{ubinter} provides the stated result of Theorem~\ref{theo:Generalproba}  .
\subsection{Proof of Proposition~\ref{cor:boundH}}
We control each deterministic term of the bound obtained in Theorems~\ref{theo:GeneralEsp} and \ref{theo:Generalproba}.
Using \eqref{pi}, we first have for any $f\in\L_2$, 
\begin{align*}
	\|\Pi_D\Gamma\Pi_D(f)-\Gamma(f)\|^2&=\int_0^1\Big(\Pi_D\Gamma\Pi_D(f)(t)-\Gamma(f)(t)\Big)^2dt\\
	&=\int_0^1\Big(\int_0^1\Pi_{S_D^2}K(s,t)f(s)ds- \int_0^1 K(s,t)f(s)ds\Big)^2dt\\
	&=\int_0^1\Big(\int_0^1(\Pi_{S_D^2}K(s,t)-K(s,t))f(s)ds\Big)^2dt\\
	&\leq\int_0^1\left[\int_0^1(\Pi_{S_D^2}K(s,t)-K(s,t))^2ds\int_0^1f^2(s)ds\right]dt\\
	&\leq\|\Pi_{S_D^2}K-K\|^2\|f\|^2
\end{align*}
and then 
\[
\vvvert\Pi_D\Gamma\Pi_D-\Gamma\vvvert^2\leq\|\Pi_{S_D^2}K-K\|^2. 
\] 
Now, we take $(s,t)\in [0,1]^2$. Then there exists a unique couple $(\lambda,\lambda')\in\Lambda_D^2$ such that $s\in I_\lambda$ and $t\in I_{\lambda'}$.
Therefore, $\phi_{\lambda''}(s)=0$ for $\lambda''\not=\lambda$ and $\phi_{\lambda'''}(t)=0$ for $\lambda'''\not=\lambda'$ and then,
\begin{eqnarray*}
	\Pi_{S_D^2}K(s,t)-K(s,t)&=&\sum_{\lambda'',\lambda'''\in\Lambda_D}\int_0^1\int_0^1 K(s',t')\phi_{\lambda''}(s')\phi_{\lambda'''}(t')ds'dt' \phi_{\lambda''}(s)\phi_{\lambda'''}(t)-K(s,t)\\
	&=&\int_0^1\int_0^1 K(s',t')\phi_{\lambda}(s')\phi_{\lambda'}(t')ds'dt' \phi_{\lambda}(s)\phi_{\lambda'}(t)-K(s,t)\\
	&=&D^2\int_{I_\lambda}\int_{I_{\lambda'}} (K(s',t')-K(s,t))ds'dt'.
\end{eqnarray*}
Then, Eq.~\eqref{Kernelregularity}   gives
\begin{align*}
	\left|\Pi_{S_D^2}K(s,t)-K(s,t)\right|&\leq D^2\sqrt{L\|K\|_\infty}\int_{I_\lambda}\int_{I_{\lambda'}}\Big[ |s'-s|^\alpha + |t-t'|^{\alpha}\Big]ds'dt'\\
	&\leq \frac{4\sqrt{L\|K\|_\infty}}{\alpha+1}D^{-\alpha},
\end{align*}
meaning that 
\[
\vvvert\Pi_D\Gamma\Pi_D-\Gamma\vvvert^2\leq \frac{16L\|K\|_\infty}{(\alpha+1)^2} D^{-2\alpha}. 
\] 
For studying the terms $A_p^{(K)}(\phi,D)$ and $A_p^{(\sigma)}(\phi,D)$, we set  for any $h=0,\ldots,p-1$, $b_h=h/p$. Observe that $t_h=h/(p-1)\in[b_h,b_{h+1}]$. We also set for any $\lambda=0,\hdots,D-1$, 
\[J_\lambda=\{h=0,\ldots,p-1:\ Leb([b_h,b_{h+1}]\cap  I_\lambda)\not=0\}.\]
Remember that $m:=p/D$ is an integer, so that, $J_\lambda=\{m\lambda,\ldots,m\lambda+m-1\}$ and
\[I_\lambda=\Big[\frac{m\lambda}{p},\frac{m\lambda+m}{p}\Big]=\bigcup_{h\in J_\lambda}[b_h,b_{h+1}].\]
Then, since $\phi_\lambda(x)=\sqrt{D}1_{I_\lambda}(x)$ and $\text{card}(J_\lambda)=m=p/D$, for any $\lambda,\lambda'=0,\hdots,D-1$, 
\begin{eqnarray*}
	G^{(K)}_{\lambda,\lambda'}&:=&\frac1{p^2}\sum_{h,h'=0}^{p-1}K(t_h,t_{h'})\phi_\lambda(t_h)\phi_{\lambda'}(t_{h'})-\int_0^1\int_0^1K(s,t)\phi_{\lambda}(s)\phi_{\lambda'}(t)dsdt\\
	&=&\sum_{h,h'=0}^{p-1}\int_{b_h}^{b_{h+1}}\int_{b_{h'}}^{b_{h'+1}}\Big[K(t_h,t_{h'})-K(s,t)\Big]\phi_{\lambda}(s)\phi_{\lambda'}(t)dsdt\\
	&&+\sum_{h,h'=0}^{p-1}\int_{b_h}^{b_{h+1}}\int_{b_{h'}}^{b_{h'+1}}K(t_h,t_{h'})\Big[\phi_{\lambda}(t_h)\phi_{\lambda'}(t_{h'})-\phi_{\lambda}(s)\phi_{\lambda'}(t)\Big]dsdt\\
	&=&D\sum_{h\in J_\lambda}\sum_{h'\in J_{\lambda'}}\int_{b_h}^{b_{h+1}}\int_{b_{h'}}^{b_{h'+1}}\Big[K(t_h,t_{h'})-K(s,t)\Big]dsdt.
\end{eqnarray*}
Therefore,
\begin{eqnarray*}
	|G^{(K)}_{\lambda,\lambda'}|&\leq&D\sum_{h\in J_\lambda}\sum_{h'\in J_{\lambda'}}\int_{b_h}^{b_{h+1}}\int_{b_{h'}}^{b_{h'+1}}\sqrt{\|K\|_\infty L}\Big(|s-t_h|^\alpha+|t-t_{h'}|^\alpha\Big)dsdt\\
	&\leq&2\sqrt{\|K\|_\infty L}\times Dp^{-1}\text{card}(J_{\lambda'})\sum_{h\in J_\lambda}\int_{b_h}^{b_{h+1}}|s-t_h|^\alpha ds\\
	&\leq&2\sqrt{\|K\|_\infty L}\times Dp^{-1}\text{card}(J_{\lambda'})\text{card}(J_{\lambda})\times \frac{2}{\alpha+1}p^{-\alpha-1}\\
	&\leq&\frac{4\sqrt{\|K\|_\infty L}}{\alpha+1}D^{-1}p^{-\alpha}.
\end{eqnarray*}
Finally,
\[
A_p^{(K)}(\phi,D)=\|G^{(K)}\|_F^2=\sum_{\lambda,\lambda'\in\Lambda_D}\left(G_{\lambda,\lambda'}^{(K)}\right)^2\leq \frac{16\|K\|_\infty L}{(\alpha+1)^2}p^{-2\alpha}.
\]
Similarly, for any $\lambda,\lambda'=0,\hdots,D-1$, observing that for $\lambda\neq\lambda'$, $J_{\lambda}\cap J_{\lambda'}=\emptyset$,
\begin{eqnarray*}
	G_{\lambda,\lambda'}^{(\sigma)}&:=&\frac{\sigma^2}{p}\left(\frac1p\sum_{h=0}^{p-1}\phi_\lambda(t_h)\phi_{\lambda'}(t_{h})-\langle\phi_\lambda,\phi_{\lambda'}\rangle\right)\\
	&=&\frac{\sigma^2}{p}\left(\frac1p\sum_{h\in J_{\lambda}\cap J_{\lambda'}}D-\1_{\{\lambda=\lambda'\}}\right)=0
\end{eqnarray*}
and
\[A_p^{(\sigma)}(\phi,D)=\|G^{(\sigma)}\|_F^2=\sum_{\lambda,\lambda'\in\Lambda_D}\left(G_{\lambda,\lambda'}^{(\sigma)}\right)^2=0.
\]
Finally, for any $\lambda=0,\hdots,D-1$, 
\begin{eqnarray*}
	\sigma_\lambda^2+s_\lambda^2&=&\frac{\sigma^2}{p^2}\sum_{h=0}^{p-1}\phi_\lambda^2(t_h)+\frac1{p^2}\sum_{h,h'=0}^{p-1}K(t_h,t_{h'})\phi_\lambda(t_h)\phi_\lambda(t_{h'})\\
	&\leq&\frac{D\sigma^2}{p^2}\text{card}(J_{\lambda})+\frac{\|K\|_\infty D}{p^2}(\text{card}(J_{\lambda}))^2\\
	&\leq& \frac{\sigma^2}{p}+\frac{\|K\|_\infty }{D}
\end{eqnarray*}
and
\begin{eqnarray*}
	\sum_{\lambda\in\Lambda_D}\Big[\sigma_\lambda^2+s_\lambda^2\Big]&\leq&\|K\|_\infty+ \frac{\sigma^2 D}{p}.
\end{eqnarray*}
This ends the proof of Proposition~\ref{cor:boundH}.
%%%%%%%%%%%%%%%%%%%%%%%%
\subsection{Proof of Lemma~\ref{lemma:A2}}\label{Appendix:A2}
To prove Lemma~\ref{lemma:A2}, we just need to prove that for all $v=(v_h)_{h=0,\ldots,p}\in\R^p$,
\begin{equation}\label{A2}
E[\exp(tv^T{\bf Z})]\leq \exp\big(ct^2 M^2E[(v^T{\bf Z})^2]\big),\quad \forall\,t\in\R,
\end{equation}
with $c$ an absolute constant (see Proposition~2.5.2 of \cite{Vershynin2018}). 
%(see (5.12) of \cite{Vershynin2012}). 
For this purpose, we denote 
$$u_{v,d}:=\sum_{h=0}^{p-1}v_h\eta_d^*(t_h)\in\R.$$ 
Then, using \eqref{eq:KL}, we have 
\begin{align*}
E[(v^T{\bf Z})^2]&=\sum_{d \in \mathbb{N}^*} \mu^{*}_du_{v,d}^2
\end{align*}
and, %still by using (5.12) of \cite{Vershynin2012}, 
still by using Proposition~2.5.2 of \cite{Vershynin2018},
$\forall\,t\in\R$,
\begin{align*}
E[\exp(tv^T{\bf Z})]&=\prod_{d\in \mathbb{N}^*}E\Big[\exp\big(t\zeta^*_d \mu^{*1/2}_du_{v,d} \big)\Big]\\
&\leq\prod_{d\in \mathbb{N}^*}E\Big[\exp\big(ct^2 \mu^{*}_du_{v,d}^2\|\zeta^*_d \|_{\psi_2}^2\big)\Big],
\end{align*}
where $c$ is an absolute constant. 
By using \eqref{M}, we obtain
\begin{align*}
E[\exp(tv^T{\bf Z})]&\leq \prod_{d\in \mathbb{N}^*}E\Big[\exp\big(ct^2 \mu^{*}_du_{v,d}^2M^2\big)\Big]\\
&\leq \exp\big(ct^2M^2E[(v^T{\bf Z})^2]\big)
\end{align*}
and \eqref{A2} is satisfied. 
\end{document}